\documentclass[conference]{IEEEtran}
\usepackage[table]{xcolor}

\usepackage{graphicx}

\ifCLASSINFOpdf
\else
\fi
\usepackage[table]{xcolor}
\usepackage{graphicx}
\usepackage{amsmath}
\usepackage{amssymb} 
\usepackage{bbm}
\usepackage{tabularx}
\usepackage{booktabs}
\usepackage{graphicx} 
\usepackage{hyperref}
\usepackage{xcolor}

\usepackage{xurl}
\newcommand{\raw}[1]{\nolinkurl{#1}}

\usepackage{amsthm}
\newtheorem{definition}{Definition}

\newenvironment{ltlformula}{\small\begin{equation*}}{\end{equation*}\ignorespacesafterend}
\usepackage{amsthm}
\newtheorem{theorem}{Theorem}

\begin{document}
%
\title{SpecGuard: Detecting Multi-Turn LLM Jailbreaks via Specification-Based Runtime Monitoring}

\author{
  \IEEEauthorblockN{Junle Li,
  Siqi Chen,
  Jiakai Li,
  Bingzhuo Zhong\textsuperscript{*}}
  \IEEEauthorblockA{\textit{Hong Kong University of Science and Technology (Guangzhou)}, Guangzhou, China\\
  Email: \{juliejunleli, bingzhuoz\}@hkust-gz.edu.cn, \{schen972, jli196\}@connect.hkust-gz.edu.cn}
}
	

%


\IEEEoverridecommandlockouts
\makeatletter\def\@IEEEpubidpullup{6.5\baselineskip}\makeatother
\IEEEpubid{\parbox{\columnwidth}{
		Network and Distributed System Security (NDSS) Symposium 2027\\
		22--26 March 2027, Seoul, Republic of Korea\\
		ISBN 978-1-970672-09-1\\  
		https://dx.doi.org/10.14722/ndss.2027.[23$|$24]xxxx\\
		www.ndss-symposium.org
}
\hspace{\columnsep}\makebox[\columnwidth]{}}

\maketitle

\begin{abstract}
Multi-turn jailbreaks turn dialogue safety into a trace property: a violation may become identifiable only when user intent and assistant behavior are considered across multiple turns. Existing guard models largely classify individual messages or exchanges and return risk scores or labels, without executable safety semantics or a formally checkable decision trail. We present SpecGuard, a training-free, specification-driven framework for monitoring multi-turn LLM interactions. SpecGuard uses a fixed atomic-proposition contract to map natural-language observations into causally scoped, evidence-backed facts, and expresses forbidden interaction patterns as Linear Temporal Logic specifications. Each specification is compiled into a symbolic finite automaton that incrementally consumes the resulting dialogue trace. The same monitoring mechanism supports buffered runtime checking and retrospective auditing, while accepting runs produce replayable witnesses that identify the triggered specification, supporting facts, and relevant dialogue turns. Experiments on multi-turn safety benchmarks show that SpecGuard consistently outperforms representative open-source guard models while maintaining controlled false-positive rates. These results demonstrate a practical path from heuristic risk classification toward formally specified, verifiable, and auditable LLM safety monitoring, without learned detector parameters or access to the protected model.
\end{abstract}


%
\IEEEpeerreviewmaketitle

 \section{Introduction}
\label{sec:Intro}

Large language models (LLMs) are increasingly deployed as interactive assistants in question answering, software development, agentic workflows, and human--AI collaboration. Despite substantial advances from alignment training and refusal mechanisms, attackers can still craft inputs that induce models to bypass their safety constraints and produce harmful content. Such attacks are commonly referred to as jailbreaks~\cite{wei2023jailbroken}. As LLMs increasingly serve real users, reliable and auditable monitoring of their interactions has become an important systems security problem.

This challenge is particularly acute in multi-turn interactions. Unlike single-turn attacks that issue a harmful request directly, attacks such as Crescendo~\cite{russinovich2024crescendo} and X-Teaming~\cite{rahman2025xteaming} progressively establish an objective, adapt the attack strategy, and elicit assistance over a sequence of turns. When a turn is examined in isolation, neither the user request nor the model reply may provide sufficient evidence of a violation. The full safety semantics emerge only when the previously established objective, subsequent continuation, and substantive assistance from the model are considered together. Multi-turn jailbreak detection is therefore not merely a text-classification problem; it requires reasoning about behavioral relations in the interaction trace.

Existing defenses include training-time safety alignment, external guard
models~\cite{inan2023llamaguard,han2024wildguard}, and inference-time
interventions~\cite{xu2024safedecoding,CircuitBreaker}. However, most guard
models classify individual prompts, responses, or turns and aggregate these
local decisions at the conversation level. This approach poorly captures
temporal behaviors such as objective continuation, gradual escalation, and
refusal bypass. Moreover, neural classifiers typically output only a risk
label or score, without identifying the satisfied safety conditions, the
turns responsible for an alarm, or a reproducible audit trail.

Formal methods offer an alternative perspective. In safety-critical systems such as autonomous driving, robotics, and task planning, safety requirements are commonly expressed as executable temporal-logic specifications, and system traces are checked against those requirements through model checking or runtime verification~\cite{Mehdipour2023FormalMethodsDriving,Alshiekh2018Shielding}. The value of this paradigm lies not only in producing a final verdict, but also in making safety conditions explicit objects that can be checked, reproduced, and maintained. This motivates our central question:

\emph{Can multi-turn jailbreak detection be formulated as specification monitoring over a dialogue trace, rather than as another black-box classifier over the entire conversation?}

Applying this paradigm to natural language is nontrivial. The challenge is not merely to choose a temporal logic, but first to determine what the formal trace should represent. Program and physical states are typically described by predefined variables, whereas harmful objectives, execution requests, and conversational strategies must be abstracted from open-ended language. The semantic granularity of atomic propositions (APs) directly affects verification. If an AP already summarizes the final safety judgment, the specification merely reproduces an upstream classifier's conclusion and cannot express the concrete linguistic facts that led to it or support interpretable cross-turn composition. Conversely, if APs record only scattered surface cues, their meaning becomes unstable under paraphrase and contextual variation, making it difficult to distinguish substantive harmful assistance from benign discussion or an effective refusal. APs should therefore represent process-level facts that can be determined from a causal prefix, checked independently, and linked to their source evidence.

AP design is coupled with the strength of the specification constraints. Overly permissive specifications may treat the local co-occurrence of a high-risk topic and an execution-oriented expression as a violation, resulting in excessive false positives. Overly restrictive specifications may cover only a small set of fixed phrasings and miss realistic variants involving paraphrase, gradual escalation, or contextual continuation. Effective safety verification must balance coverage and discrimination by combining complementary evidence, relevant temporal relations, and necessary assistant-side behavior. The central design problem is thus not to define APs or temporal formulas independently, but to jointly determine the semantic boundaries of process-level facts, their causal activation conditions, and the scope of the specifications that constrain them.

To handle the difficulties above, we present SpecGuard, a specification-driven safety monitoring framework for multi-turn LLM interactions. Its core is an auditable AP contract that excludes result-shaped facts. We derive safety-relevant observations from security policies, known jailbreak behaviors, and development dialogues, and organize them into 15 canonical APs with explicit speaker roles, causal inputs, and semantic boundaries. These APs describe independently checkable process-level facts, such as whether the user establishes a harmful objective or issues an execution request, whether the assistant provides substantive content or procedural detail, and whether the response constitutes an effective refusal. Each activated AP retains its raw evidence provenance, making the conversion from natural language to symbolic traces inspectable and replayable.

On top of this contract, SpecGuard encodes forbidden behaviors as ten Linear Temporal Logic (LTL) specifications: four same-turn specifications and six adjacent cross-turn specifications. Each specification requires substantive assistant-side evidence; a harmful request or attack signal alone cannot trigger a violation. SpecGuard compiles each specification into a symbolic deterministic finite automaton whose transition guards operate over canonical AP valuations. These automata consume the dialogue trace incrementally, converting the declarative safety requirements into executable finite-state monitors.

SpecGuard supports two execution modes using the same compiled automata. In runtime mode, a candidate response is held in a private buffer while the monitor advances the automata over the newly completed dialogue transaction. The response is released only if no specification reaches an accepting state. In audit mode, the system reconstructs an existing dialogue trace in causal order and records the first accepting transition for each matched specification. An accepting run produces a structured witness containing the triggered specification, relevant turns, matched APs, and their supporting evidence. The framework contains no learned detector parameters and requires neither access to the target LLM's weights, gradients, or internal activations nor modification of the protected model, allowing it to operate as an external monitor for any black-box LLM service.

To show the effectiveness of SpecGuard, we evaluate it on multiple multi-turn safety benchmarks and compare it with representative open-source guard models under identical inputs. SpecGuard achieves stronger detection performance across languages and attack trajectories while keeping false positives under control. Additional auditing shows that its alarms can be reproduced from the corresponding specification witnesses and that the runtime and audit paths agree when applied to the same dialogue under the shared semantics.

Our contributions are as follows:

\begin{itemize}
\item \textbf{AP contract for conversational safety.}
 We introduce an AP contract that connects natural-language analysis with formal verification by specifying which safety-relevant facts may enter a dialogue trace, when they may be causally observed, and how their source evidence is retained. The contract converts open-ended language into canonical propositions with explicit semantic boundaries, speaker attribution, and provenance. Beyond the resulting contract, the methodology itself is also transferable, providing a blueprint for constructing goal-oriented AP libraries.

\item \textbf{Data-grounded temporal safety specifications.}
We translate safety policies, known jailbreak behaviors, and interaction patterns observed in development dialogues into typed LTL specifications covering both same-turn violations and cross-turn behavioral evolution. This construction provides coverage across attack variants while explicitly controlling the scope of each constraint, making safety requirements executable, checkable, and maintainable. More broadly, the pipeline provides a transferable methodology for converting open-world knowledge into formal constraints for new monitoring scenarios.

\item \textbf{Finite-state monitoring framework.}
We compile typed LTL specifications into symbolic finite automata that incrementally consume causal dialogue traces. The resulting monitor supports both buffered runtime filtering and retrospective auditing, while accepting runs produce replayable witnesses that connect each verdict to the responsible specification, dialogue turns, and supporting AP evidence.

\item \textbf{Formal safety methodology for LLM jailbreak defense.}
We systematically show how runtime verification can be applied to black-box LLM safety defense by converting natural-language safety requirements into checkable specifications and coupling formal verdicts with auditable evidence. This provides a general path from heuristic moderation to specification-level checkability without modifying or retraining the protected model.
\end{itemize}

\section{Related Work}\label{section:Related Work}

\subsection{Multi-Turn Jailbreak Attacks and Existing Defenses}\label{subsection:RL work1}
Jailbreak attacks elicit policy-violating outputs from aligned language models through adversarial prompting~\cite{wei2023jailbroken}, with automated techniques ranging from universal adversarial suffixes~\cite{zou2023universal} to black-box iterative refinement~\cite{chao2023pair}. Recent attacks have shifted decisively toward multi-turn forms: Crescendo~\cite{russinovich2024crescendo}, Many-shot Jailbreaking~\cite{anil2024manyshot}, ActorAttack~\cite{ren2024actorattack}, and X-Teaming~\cite{rahman2025xteaming} erode a model's refusal behavior through sequences of individually innocuous turns, so that the malicious intent is distributed across the dialogue trace and resides in no single message.

On the defense side, safety judgment has traditionally been cast at the granularity of individual messages: a long line of guard models---from Llama~Guard~\cite{inan2023llamaguard} to recent bilingual variants~\cite{han2024wildguard, zeng2024shieldgemmagenerativeaicontent, GuardReasoner}---classifies a single prompt or response, and even when the full history is provided as input, their verdicts do not reference the temporal structure between turns. Recognizing the multi-turn threat, recent work has begun to reason over dialogues explicitly: SafeDream~\cite{safedream} accumulates a learned safety estimate across turns, THRD~\cite{ma2026thrd} fuses per-turn risk with history signals through training-free evaluators, G-Guard~\cite{huang2025gguard} models cross-turn intent aggregation over graphs, and TurnGate~\cite{shen2026turngate} locates the earliest turn at which harmful capability closure occurs. These advances establish the necessity of cross-turn temporal reasoning. Across all of them, however, the criterion for violation remains \emph{implicit}: it lives in fine-tuned weights, judge prompts, or scoring heuristics, and is nowhere stated as an object with defined semantics---one that can be inspected, composed, or checked in its own right.

A parallel line of work systematizes jailbreak strategies and safety policies: Shen et al.~\cite{shen2024dan} characterize and categorize in-the-wild jailbreak prompts at scale, Liu et al.~\cite{liu2024hitchhiker} distill recurring attack patterns in prompt engineering, and on the policy side, taxonomies derived from platform regulations organize safety risks into structured categories~\cite{zeng2024airbench}. These taxonomies make attack patterns and policy constraints explicit for the first time; yet they remain descriptive---a taxonomy by itself is
neither executable nor checkable, and bears no formal connection to a running conversation. In sum, existing defenses are executable yet implicit, while existing taxonomies are explicit yet non-executable; to our
knowledge, no prior work defines multi-turn violation as a temporal property that is both explicitly stated and machine-checkable---the gap our specification layer addresses.

 \subsection{Formal Specification and Runtime Verification for Safety}\label{subsection:RL2}

Formal methods have long supplied enforceable safety guarantees in safety-critical systems: in autonomous driving~\cite{Mehdipour2023FormalMethodsDriving,Wongpiromsarn2023FormalMethodsAutonomousSystems}, robot
control~\cite{Desai2017SafeRobots}, and multi-agent planning~\cite{sun2022multiagentmotionplanningsignal}, temporal logic and runtime verification encode safety constraints and rule out violating behaviors during execution. In LLM-based systems, VeriPlan~\cite{Lee2025VeriPlan} integrates model checking into planning pipelines, ShieldAgent~\cite{chen2025ShieldAgent} constrains agent actions with verifiable rules extracted from policy documents, and R2-Guard~\cite{r2} performs probabilistic logical reasoning over the safety categories of a single input within a guardrail. Across this line of work, the system under monitoring has invariably been the action layer of a robot, a planner, or an agent; the conversational behavior of an LLM at inference time has never itself been the object of runtime verification.

A complementary line of research translates natural-language requirements and safety rules into machine-checkable formal specifications~\cite{cosler2023nl2spec, chen2024nl2tl}, with growing attention to the semantic fidelity and safety of the generated formulas~\cite{li2026AutoSafeLTL}. Its value lies in turning rules that are otherwise ambiguous, open-ended, and interpretation-dependent into specification objects that can be represented, analyzed, and reasoned about---safety requirements become checkable, composable, and enforceable only once formalized~\cite{FERRARI2025Formalrequirements,Zhang2024LLMFMRoadmap}. This line, however, rests on a tacit premise: the rules to be formalized already exist somewhere---in traffic law, requirements documents, or physical safety constraints. Dialogue safety does not satisfy this premise. Which conversational patterns constitute a jailbreak has never been written down as rules, let alone formalized; the attack taxonomies of \S\ref{subsection:RL work1} categorize jailbreaks descriptively, but a category is not a constraint---it prescribes nothing. The specification layer for dialogue safety must therefore be built from scratch, and its construction is itself an exercise in balance: rules that
are too broad penalize benign conversations on sensitive topics, while rules that are too narrow miss attacks staged through innocuous turns. To our knowledge, no formal specification set for multi-turn jailbreaks existed prior to this work; the rule base of SpecGuard is a first systematic attempt to fill this void, and its design principles are presented in Section~\ref{sec:SpecGuard Design}.

\section{Threat Model and Problem Formulation}\label{section:Threat Model and Problem Formulation}

\subsection{Threat Model}
\label{subsec:threatmodel}

\noindent\textbf{Adversary.}
We consider an adaptive adversary who interacts with an aligned LLM through
natural-language dialogue and seeks to elicit unsafe assistance. The adversary
controls the user messages and observes the model's replies. It may issue a
direct single-turn request or conduct a multi-turn jailbreak, distributing its
objective across individually innocuous-looking messages and choosing later
messages in response to the preceding dialogue.

\smallskip
\noindent\textbf{Defender.}
The defender is an external, causal monitor of the dialogue stream. After each
completed turn, it observes only the transcript available at that point and
checks the resulting trace against a fixed set of safety specifications. It
has no access to the target model's weights, gradients, or activations, and
does not modify or retrain the target model.

The monitor determines whether the observed dialogue has instantiated a
forbidden pattern. It returns a binary verdict and, for each match, identifies
the corresponding specification and a structured witness consisting of the
relevant turns and propositions. Because the monitoring interface depends only on the
dialogue transcript, it is independent of the target model's architecture.

\subsection{Preliminaries: Finite Traces and $\mathrm{LTL}_f$ Monitoring}
\label{subsec:preliminaries}

\smallskip
\noindent\textbf{Finite traces.}
To specify temporal safety properties over dialogue observations, we first fix the syntax and finite-trace semantics of $\mathrm{LTL}_f$.
\begin{definition}
\label{def:ltl}
Let $\mathcal{A}$ be a set of atomic propositions. \emph{Linear Temporal Logic}
(LTL) is a formal language for specifying temporal properties over ordered
traces. Its formulas are generated by
\[
\varphi ::= p \mid \neg\varphi \mid \varphi\land\varphi
\mid \mathbf{X}\varphi \mid \mathbf{F}\varphi,
\qquad p\in\mathcal{A},
\]
with the remaining Boolean and temporal operators defined as abbreviations.
\end{definition}

Each proposition denotes a safety-relevant fact that can be evaluated at an
observation position. A finite trace is a sequence
\[
w=A_0A_1\cdots A_T,
\qquad A_t\subseteq\mathcal{A},
\]
where $A_t$ is the set of propositions that hold at position $t$. In our
setting, each position corresponds to one observable dialogue transaction.
As new transactions become available, the monitor receives the corresponding
trace positions incrementally.

$\mathrm{LTL}_f$ interprets temporal-logic formulas over finite traces. For a
proposition $p\in\mathcal{A}$ and a position $i\in\{0,\ldots,T\}$,
\[
w,i\models_f p
\quad\Longleftrightarrow\quad
p\in A_i.
\]
We use the next operator $\mathbf{X}$ and the eventually operator $\mathbf{F}$:
\[
w,i\models_f\mathbf{X}\varphi
\quad\Longleftrightarrow\quad
i<T\ \land\ w,i+1\models_f\varphi,
\]
\[
w,i\models_f\mathbf{F}\varphi
\quad\Longleftrightarrow\quad
\exists j\in\{i,\ldots,T\}:w,j\models_f\varphi.
\]
We write $w\models_f\varphi$ as shorthand for $w,0\models_f\varphi$. For
example, $\mathbf{F}(p\land\mathbf{X}q)$ holds exactly when there exists a
position $i<T$ such that $p\in A_i$ and $q\in A_{i+1}$; the pair $(i,i+1)$
is a witness for the match. Unlike standard LTL over infinite executions,
$\mathrm{LTL}_f$ evaluates the finite trace observed so far and therefore
requires no assumptions about future dialogue turns.

\smallskip
\noindent\textbf{Finite-state monitoring.}\label{def:BA}
To evaluate these properties incrementally as dialogue turns arrive, we next introduce deterministic finite automata over AP valuations.
\begin{definition}
\label{def:dfa}
A \emph{deterministic finite automaton} (DFA)~\cite{sipser2013introduction}
over AP valuations is a tuple
\[
\mathcal{M}=(Q,2^{\mathcal{A}},\delta,q_0,F),
\]
where $Q$ is a finite set of states, $2^{\mathcal{A}}$ is the input alphabet,
$\delta:Q\times 2^{\mathcal{A}}\rightarrow Q$ is the transition function,
$q_0\in Q$ is the initial state, and $F\subseteq Q$ is the set of accepting
states. Extending \(\delta\) to finite traces in the standard way,
\(\mathcal{M}\) accepts a trace \(w\) iff
\[
\delta^\ast(q_0,w)\in F.
\]
\end{definition}

Every $\mathrm{LTL}_f$ formula defines a regular language of finite traces and
can therefore be compiled into a DFA. Given a formula \(\varphi\), let
\[
\mathcal{M}_{\varphi}=\mathsf{compile}(\varphi)
\]
denote the resulting automaton. Compilation preserves finite-trace semantics:
\[
w\models_f\varphi
\quad\Longleftrightarrow\quad
\delta^\ast(q_0,w)\in F
\quad\Longleftrightarrow\quad
\mathcal{M}_{\varphi}\text{ accepts }w.
\]
Operationally, the monitor consumes one valuation \(A_t\) at a time and updates
the current automaton state through \(\delta\). An accepting state after
consuming the current trace indicates that the trace satisfies \(\varphi\).
For the forbidden-pattern specifications used in this work, accepting states
are absorbing, so acceptance persists under trace extension, while the
accepting run identifies the positions witnessing the match.

\subsection{Problem Formulation}
\label{subsec:problem-formulation}

We formalize the detection problem of Sec.~\ref{subsec:threatmodel} in two parts: a \emph{specification} layer that fixes, in advance, which dialogue behaviors are forbidden, and a \emph{monitoring} layer that evaluates these specifications on observed dialogues. Safety-relevant facts of the conversation are abstracted into atomic propositions (APs), specifications are temporal-logic formulas over these propositions, and monitoring reduces to evaluating the formulas on the trace a dialogue induces.

\noindent\textbf{Dialogue traces.}
Let $D=(u_0,r_0,\ldots,u_T,r_T)$ be a finite dialogue, where
$u_t$ and $r_t$ denote the user message and model reply at turn $t$.
Let
$P_t^{-}=(u_0,r_0,\ldots,r_{t-1},u_t)$
be the pre-reply prefix and
$P_t^{+}=P_t^{-}\circ r_t$
the corresponding post-reply prefix.

The AP universe is partitioned as
\vspace{-4pt}
\[
\mathcal{A}
=
\mathcal{A}_{\mathrm{U}}
\mathbin{\dot\cup}
\mathcal{A}_{\mathrm{R}}
\mathbin{\dot\cup}
\mathcal{A}_{\mathrm{D}},
\]
where propositions in $\mathcal{A}_{\mathrm{U}}$ describe the current
user message, propositions in $\mathcal{A}_{\mathrm{R}}$ describe the
current reply, and propositions in $\mathcal{A}_{\mathrm{D}}$ summarize
the dialogue history, such as whether an objective established earlier
remains active.

Let $E_{\mathrm{U}}$, $E_{\mathrm{R}}$, and $E_{\mathrm{D}}$ be causal
extractors for these three families. For
$\mathrm{X}\in\{\mathrm{U},\mathrm{R},\mathrm{D}\}$,
$E_{\mathrm{X}}(P)\subseteq\mathcal{A}_{\mathrm{X}}$ denotes the
propositions activated by the observed prefix $P$. The valuation at turn $t$ is
\vspace{-4pt}
\[
A_t
=
E_{\mathrm{U}}(P_t^{-})
\cup
E_{\mathrm{R}}(P_t^{+})
\cup
E_{\mathrm{D}}(P_t^{+}).
\]

Thus, user-side propositions do not depend on the reply they precede,
and dialogues sharing an observed prefix induce the same trace prefix.
The dialogue induces the finite trace
$w(D)=A_0A_1\cdots A_T$, where $A_t\in 2^{\mathcal{A}}$ records the
observations available at turn $t$.
The concrete AP inventory and extraction procedures are described in
Sec.~\ref{sec:SpecGuard Design}.

\smallskip
\noindent\textbf{Safety specifications.}
A specification set
$\Phi=\{\phi_1,\ldots,\phi_m\}$
collects forbidden dialogue patterns.
Each $\phi_i$ is an $\mathrm{LTL}_f$ formula over $\mathcal{A}$,
interpreted over finite traces under $\models_f$.
Temporal operators range over dialogue turns, allowing a formula to relate
observations within or across turns.
Every matched pattern requires multiple observations, including substantive
assistant-side evidence; user behavior alone is insufficient to establish a
violation.
Relevant dialogue history is represented by propositions in
$\mathcal{A}_{\mathrm D}$, allowing all specifications to remain in the
$\{\mathbf{F},\mathbf{X}\}$ fragment of $\mathrm{LTL}_f$.
Each specification is compiled once into a finite-state monitor that preserves
its $\mathrm{LTL}_f$ semantics.
The AP contract, its deterministic activators, and the specification set are
fixed before evaluation.

\smallskip
\noindent\textbf{Monitoring task.}
Given a dialogue $D$ and its trace $w(D)$, the monitor raises an alarm if the
trace satisfies at least one specification:
\vspace{-4pt}
\[
\mathsf{Alarm}(D)
=
\mathbbm{1}\!\left[
\exists\,\phi\in\Phi:\;w(D)\models_f\phi
\right].
\]
Equivalently, an alarm is raised when the finite-state monitor compiled from
some $\phi\in\Phi$ accepts $w(D)$, indicating that the corresponding conversation violates the safety rule specified by $\phi$. For each matched specification, the monitor
reports an accepting-run witness identifying the dialogue turns and
propositions that realize the match.
Sec.~\ref{subsec:specguard-monitoring} describes this process in runtime and
audit modes.
\section{SpecGuard Design}\label{sec:SpecGuard Design}

\subsection{Framework Overview}
\label{subsec:FrameworkOverview}

\begin{figure*}[htbp] 
    \centering
    \includegraphics[width=1\textwidth]{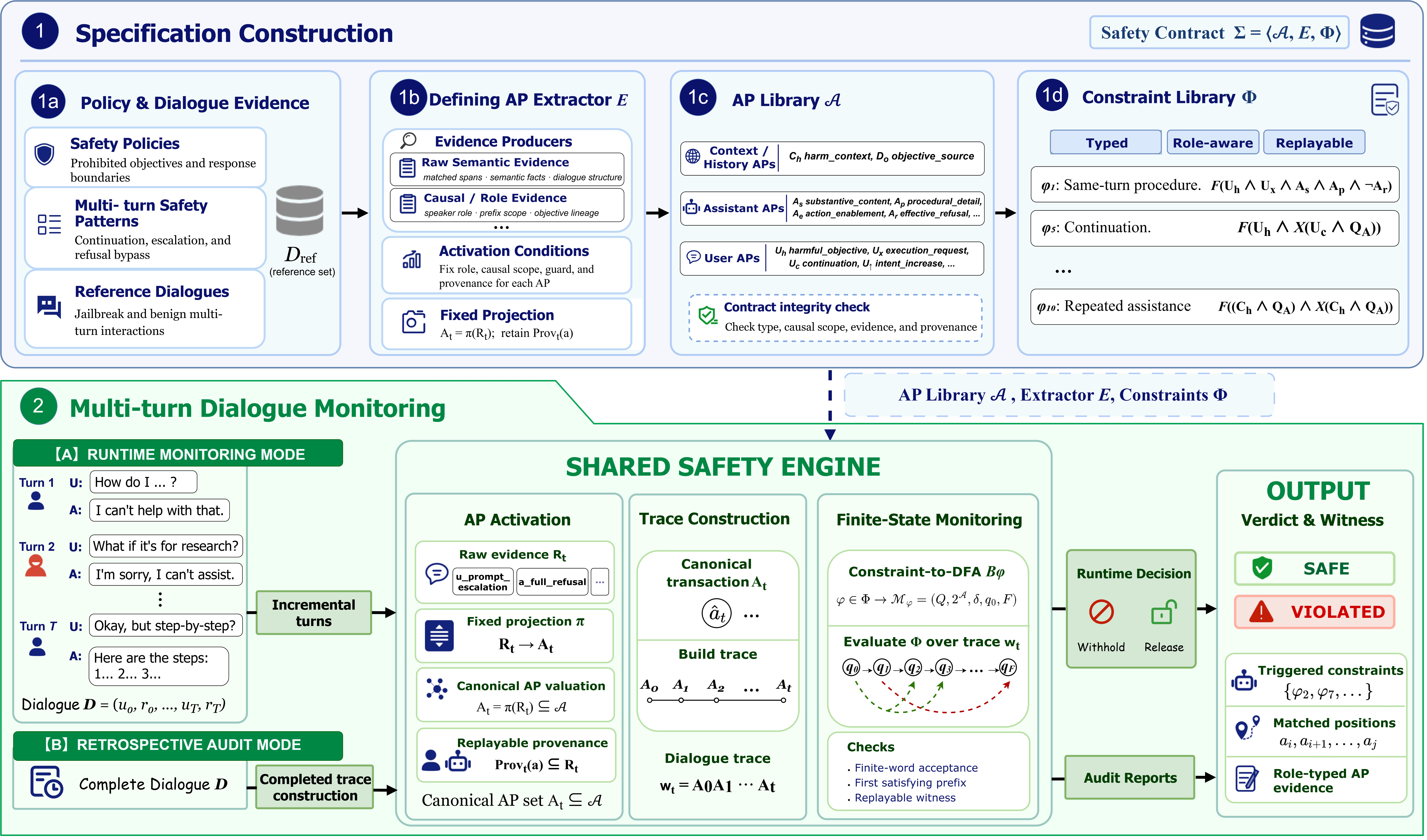} 
    \caption{SpecGuard Overview}
    \label{fig:specguard}
\end{figure*}

Figure~\ref{fig:specguard} presents SpecGuard's architecture. The
\emph{Specification} component defines a canonical AP contract and a library
of \(\mathrm{LTL}_f\) safety specifications. The \emph{Monitoring} component
extracts AP valuations from the dialogue, constructs a turn-level trace, and
checks it for specified forbidden patterns.

Monitoring supports two execution modes over the same formal representation.
\emph{Runtime mode} updates the verdict incrementally after each completed
turn. \emph{Audit mode} replays a completed dialogue and returns the matched
specifications together with structured witnesses. Both modes evaluate the
dialogue under the same temporal semantics.

The central design principle is a separation between fact abstraction
and violation judgment. Upstream components populate the AP contract with
canonical facts, while the final verdict is determined solely by matching the
resulting trace against the temporal specifications. We next describe the
construction of the Specification component, followed by the runtime and audit
modes of Monitoring.

\subsection{Specification Construction}
\label{subsec:specguard-specification}

Specification construction compiles safety requirements into
an executable and auditable formal contract
\vspace{-4pt}
\[
\Sigma=\langle\mathcal{A},E,\Phi\rangle,
\]
where $\mathcal{A}$ is a vocabulary of canonical atomic propositions (APs),
$E$ causally instantiates these APs from dialogue evidence, and $\Phi$ is a
library of typed $\mathrm{LTL}_f$ constraints defining forbidden relations.
Construction must answer two questions: which dialogue facts merit formal
representation, and under what relations those facts constitute a violation.
SpecGuard follows a \emph{semantics-first} process: it decomposes
safety requirements into observable factors, validates their instantiation
against contrastive dialogue evidence, and only then composes them into
temporal constraints. This ordering prevents benchmark labels or detector
outputs from being restated as formal predicates.

\begin{table*}[t]
\centering
\footnotesize
\caption{Complete canonical AP vocabulary. The identifier column gives the
exact name emitted by the implementation.}
\label{tab:app-ap-definitions}
\setlength{\tabcolsep}{4pt}
\renewcommand{\arraystretch}{1.08}
\begin{tabularx}{\textwidth}{@{}lp{0.08\textwidth}p{0.25\textwidth}X@{}}
\toprule
\textbf{Family} & \textbf{Symbol} & \textbf{Canonical identifier}
& \textbf{Observable semantics} \\
\midrule
Derived & $C_h$ & \texttt{c\_harm\_context}
& The active dialogue context concerns a harmful objective. \\
Derived & $D_o$ & \texttt{d\_objective\_source}
& An observed earlier turn established the objective referenced later. \\
\midrule
User & $U_h$ & \texttt{u\_harmful\_objective}
& The current user message states or pursues a concrete harmful objective. \\
User & $U_x$ & \texttt{u\_execution\_request}
& The user asks for execution, operational parameters, or actionable steps. \\
User & $U_e$ & \texttt{u\_evasion\_tactic}
& The request uses an evasion, encoding, persona, bypass, or inversion tactic. \\
User & $U_c$ & \texttt{u\_continuation}
& The user asks to continue, complete, or reuse preceding dialogue content. \\
User & $U_{\uparrow}$ & \texttt{u\_intent\_increase}
& The user increases the specificity, commitment, or operationality of the request. \\
User & $U_o$ & \texttt{u\_objective\_reference}
& The user explicitly refers to an objective established in the observed history. \\
\midrule
Reply & $A_s$ & \texttt{a\_substantive\_content}
& The reply contains substantive content rather than an empty or generic response. \\
Reply & $A_p$ & \texttt{a\_procedural\_detail}
& The reply supplies procedural or otherwise concrete harmful detail. \\
Reply & $A_m$ & \texttt{a\_method\_specificity}
& The reply names a specific method, mechanism, or operational choice. \\
Reply & $A_e$ & \texttt{a\_action\_enablement}
& The reply enables action through a workaround, artifact, or feasibility claim. \\
Reply & $A_t$ & \texttt{a\_targeted\_content}
& The reply supplies harmful content specialized to a target, claim, or group. \\
Reply & $A_{\ell}$ & \texttt{a\_partial\_assistance}
& The reply provides incomplete assistance that may support later escalation. \\
Reply & $A_r$ & \texttt{a\_effective\_refusal}
& The reply withholds the requested harmful assistance through an effective refusal. \\
\bottomrule
\end{tabularx}
\end{table*}

\smallskip
\noindent\textbf{Complementary sources and development data.}
SpecGuard's specification and AP contract draw on three complementary
sources. First, existing safety policies and content-risk taxonomies define
the semantic boundaries to protect, such as the prohibition against providing
actionable methods or substantive assistance for harmful objectives
\cite{markov2023holistic,inan2023llamaguard,mazeika2024harmbench,han2024wildguard}.
Second, prior studies of multi-turn jailbreaks identify temporal relations
that the representation must capture, including request continuation, intent
escalation, objective reference, and refusal bypass
\cite{russinovich2024crescendo,shen2026turngate}.
Third, we use jailbreak, refusal, and benign interaction examples from
WildChat-1M~\cite{zhao2024wildchat},
MTID-train~\cite{shen2026turngate}, and
WildGuardMix~\cite{han2024wildguard} as development data. This corpus is used
to author and validate the bilingual evidence templates, deterministic
relation parsers, and contract guards, and to test whether candidate facts are
causally observable and distinguish harmful assistance from effective refusal
or legitimate safety discussion. It is not used to induce APs directly from
endpoint labels. Together, the first two sources determine what must be
represented, while the development corpus supports its operationalization.

\smallskip
\noindent\textbf{From safety requirements to canonical APs.}
We decompose each safety requirement into minimal semantic factors that can be
observed independently but do not determine a verdict in isolation. A candidate
factor enters the AP contract only if it (i) can be instantiated causally from
the observed dialogue prefix, (ii) denotes an observation rather than a
verdict, (iii) separates at least one pair of semantically similar dialogues
with different safety implications, and (iv) can be reused across
specifications while retaining its supporting evidence.

For example, the requirement that a model must not provide concrete methods
for a harmful execution request is not represented by a single
\texttt{harmful\_delivery} proposition. Instead, we separate the user-side
facts \emph{harmful objective} and \emph{execution request} from the
assistant-side facts \emph{substantive content}, \emph{procedural detail}, and
\emph{effective refusal}. The same request may therefore be refused, answered
with high-level safety information, or fulfilled with actionable details; only
the specified relation among these observations constitutes a violation.
Similarly, phrase-specific expressions such as ``continue the previous
response'' are normalized into a reusable \emph{continuation} factor.

The resulting vocabulary contains 15 typed APs: two context/history factors in
$\mathcal{A}_{\mathrm{D}}$, six user-side factors in
$\mathcal{A}_{\mathrm{U}}$, and seven assistant-side factors in
$\mathcal{A}_{\mathrm{R}}$. Their complete definitions appear in
Table~\ref{tab:app-ap-definitions}. Let $E_{\mathrm{raw}}$ denote the
family-specific evidence producers. The canonical extractor is
\vspace{-9pt}
\[
E=\pi\circ E_{\mathrm{raw}},
\]
where the fixed projection $\pi$ maps fine-grained evidence to canonical APs
and excludes result-shaped predicates. For every active AP $a$ at turn $t$,
SpecGuard retains its raw support as $\operatorname{Prov}_t(a)$. The complete
raw-to-canonical mapping is reported in Appendix~\ref{app:ap-contract}.

\smallskip
\noindent\textbf{From safety requirements to temporal constraints.}
Canonical APs describe individual observations, whereas a safety requirement
specifies how those observations must be related. We therefore translate each
requirement into a typed relation among context, user, and assistant facts. The
relation identifies the evidence required for a match, acceptable alternative
evidence, and any disqualifying condition, such as an effective refusal. It
also specifies whether the relevant events must occur within the same dialogue
round or across two consecutive rounds. These relations are then
deterministically expressed in $\mathrm{LTL}_f$, using $\mathbf{F}$ to locate a
forbidden pattern and $\mathbf{X}$ to enforce next-round ordering. This
construction preserves both speaker roles and temporal alignment.

Several safety requirements share the same assistant-side condition: the reply
must contain substantive material that meaningfully supports the requested
objective. To avoid repeating this condition in every constraint, we define
\begin{ltlformula}
\begin{aligned}
Q_{\mathrm{a}}
\triangleq\;&
\mathsf{a\_substantive\_content} \\
&{}\land\bigl(
   \mathsf{a\_procedural\_detail}
   \lor \mathsf{a\_method\_specificity} \\
&\qquad\ \ \lor \mathsf{a\_action\_enablement}
   \lor \mathsf{a\_targeted\_content}
\bigr) \\
&{}\land \neg\mathsf{a\_effective\_refusal}.
\end{aligned}
\end{ltlformula}
Intuitively, $Q_{\mathrm{a}}$ holds when the reply contains substantive
content together with at least one form of policy-relevant detail---procedural
steps, a specific method, action-enabling information, or targeted
content---and does not constitute an effective refusal. It is only shorthand
for this recurring combination of canonical APs; it neither introduces a new
AP nor determines a violation by itself.

We first consider a direct harmful request and its reply within the same
dialogue round:
\begin{ltlformula}
\begin{aligned}
\phi_{\mathrm{proc}}
\triangleq
\mathbf{F}\bigl(
&\mathsf{u\_harmful\_objective}
\land \mathsf{u\_execution\_request} \\
&{}\land \mathsf{a\_substantive\_content}
\land \mathsf{a\_procedural\_detail} \\
&{}\land \neg\mathsf{a\_effective\_refusal}
\bigr).
\end{aligned}
\end{ltlformula}
This formula states that there exists a round in which the user pursues a
harmful objective and requests an execution method, while the corresponding
reply provides substantive procedural content without effectively refusing.
Because all APs are evaluated at the same trace position, evidence from an
unrelated reply cannot satisfy the constraint.

Other violations emerge only when evidence is combined across rounds.
Consider a user who establishes a harmful objective and asks the model to
continue in the next round:
\vspace{-4pt}
\begin{ltlformula}
\begin{aligned}
\phi_{\mathrm{cont}}
\triangleq
\mathbf{F}\Bigl(
&\mathsf{u\_harmful\_objective} \\
&{}\land
\mathbf{X}\bigl(
\mathsf{u\_continuation}
\land Q_{\mathrm{a}}
\bigr)
\Bigr).
\end{aligned}
\end{ltlformula}
This formula requires some round $t$ to establish the harmful objective,
followed immediately by a round $t+1$ in which the user issues a continuation
request and the corresponding reply satisfies $Q_{\mathrm{a}}$. The next
operator $\mathbf{X}$ enforces this adjacency, preventing assistance from an
unrelated later exchange from being attributed to the earlier objective.

The complete specification contains four same-round constraints and six
adjacent-round constraints. Longer-range dependencies are represented by
provenance-preserving history APs, allowing constraints to reference an
established objective without introducing unconstrained links between
arbitrary turns. Appendix~\ref{app:constraint-library} provides the complete
constraint inventory and symbolic expansions.

\smallskip
\noindent\textbf{Transferability.}
The AP contract and temporal specifications of SpecGuard are independent of dataset labels and protected-model internals, allowing them to be reused across datasets and target LLMs. More broadly, the whole pipeline for deriving observable APs from open-world knowledge and constructing validated temporal constraints is transferable to new monitoring scenarios.

\smallskip
\noindent\textbf{Validation and freeze.}
Before evaluation, we validate causal observability, speaker typing, refusal
disqualification, and cross-turn ordering on harmful, benign, and refusal
cases. The final specification enforces three invariants: no single AP
determines a verdict, every constraint requires positive assistant-side
evidence, and every match yields a replayable witness grounded in raw evidence.
We additionally verify that injecting an excluded result-shaped predicate
cannot alter either the canonical trace or the monitor verdict. The AP
vocabulary, projection, and constraint library are then frozen jointly before
evaluation.

\subsection{Dialogue Monitoring}
\label{subsec:specguard-monitoring}

The preceding subsection fixes the AP contract \(\mathcal{A}\) and the
set \(\Phi\) of typed \(\mathrm{LTL}_f\) specifications. Dialogue
monitoring neither introduces new proposition semantics nor modifies
the specifications using evaluation data. Instead, it converts each
completed user--assistant interaction into an auditable AP valuation
\(A_t\) and feeds the resulting valuations to the finite-state monitors
compiled from the safety specifications.


\begin{figure*}[t]
    \centering
    \includegraphics[
        width=\textwidth,
        trim=0 0 0 0,
        clip
    ]{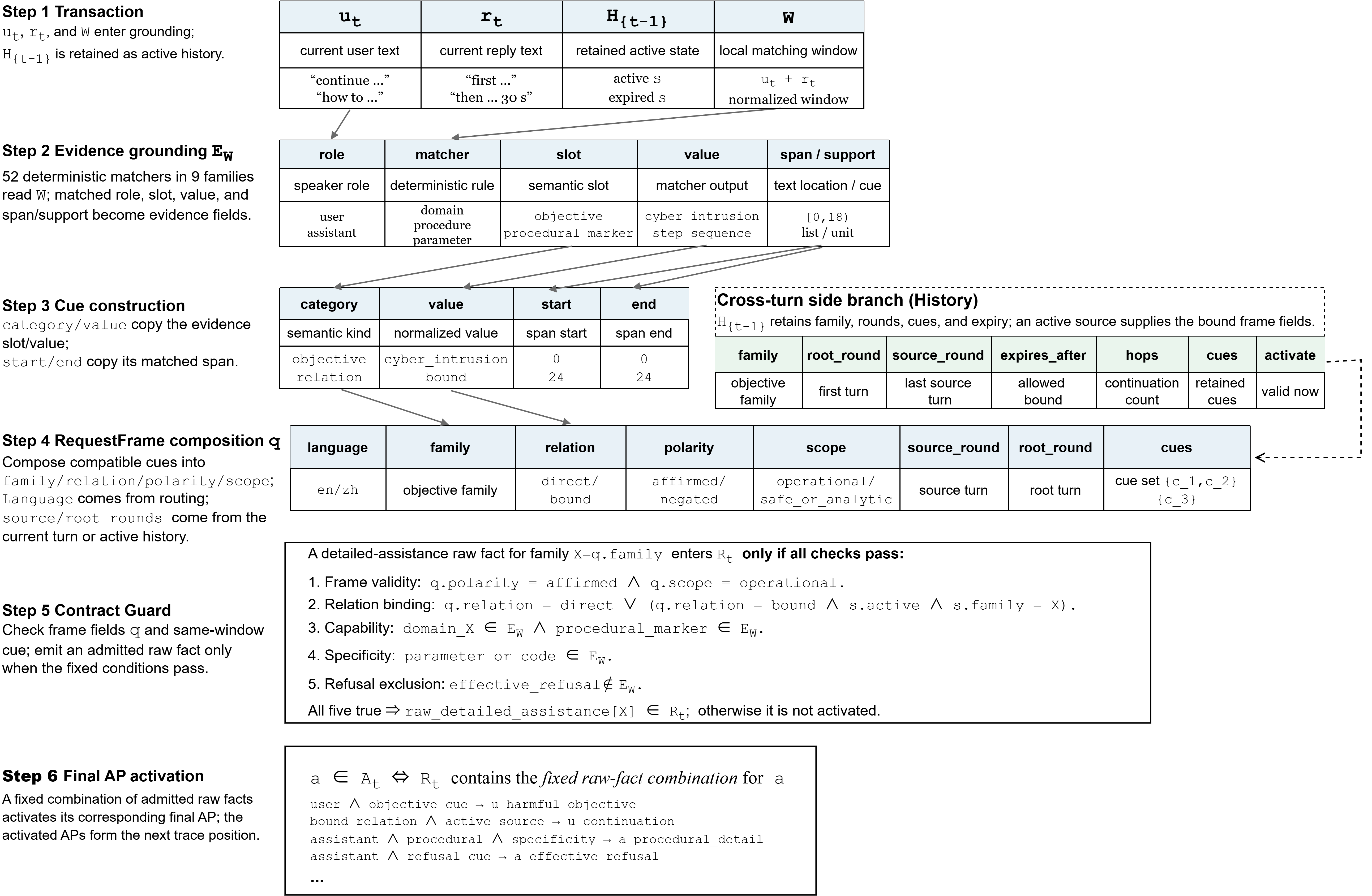}
    \caption{The AP contract and activation pipeline of \texttt{SpecGuard}. Arrows indicate data flow.}
    \label{fig:ap-activation-pipeline}
\end{figure*}

\noindent\textbf{AP Activation} converts the causal dialogue prefix observed at the current
round into the canonical AP valuation used by temporal monitoring. As shown in
Figure~\ref{fig:ap-activation-pipeline}, the process proceeds through three
stages: evidence grounding, relation composition, and contract guarding. The
first two stages construct intermediate structured facts; only raw facts that
pass the guard can activate final APs and enter the current-round valuation
$A_t$. Thus, AP activation does not assign a risk label to an entire dialogue.
Instead, it incrementally checks local evidence, request relations, and
contract conditions.

\emph{Evidence Grounding.}
At round $t$, the system extracts traceable textual facts from the current user
message, model reply, retained history, and local content window $W$.
Text normalization, language routing, and command-zone localization for long
messages only determine the scope in which matchers operate; their
implementation is described in Appendix~\ref{app:evidence-grounding}. The
current implementation contains 52 deterministic text matchers organized into
9 evidence families. Six additional bounded-relation matchers are used in the
subsequent relation-composition stage. Representative matchers include
domain, procedural/action, parameter/code, and request/relation matchers; the
complete inventory and triggering conditions are provided in
Appendix~\ref{app:evidence-grounding}.

Each matcher independently emits a provenance-bearing evidence record rather
than a violation label. A record contains the text role, matcher, semantic
slot, matched value, and match location; a single window may produce multiple
evidence records, each retaining its textual span or structural support. As
illustrated in Figure~\ref{fig:ap-activation-pipeline}, these fields are
derived from $u_t$, $r_t$, and $W$. All subsequent stages operate on the
structured evidence records and do not rescan the raw text.

\emph{Relation Composition.}
Relation composition converts evidence records into cues and combines
compatible cues into typed request frames. A cue retains the semantic
category, value, and character span of its supporting evidence. A frame
organizes the resulting information into language, risk family, request
relation, polarity, scope, source round, root round, and the cues used in its
construction, as shown in the figure.

For a request stated directly in the current round, the frame relation is
\texttt{direct}, and its family is obtained from the current objective
evidence. For a request that continues or refers to a previous objective, the
relation is \texttt{bound}; a frame is created only when the history contains
an active objective source from the same risk family. This source supplies the
historical rounds and relevant cues for the bound frame. If the required
relation, a compatible family, or a valid source is absent, no frame is
created. Relation composition therefore organizes local evidence into bounded
request facts, but does not itself activate the final APs.

\emph{Contract Guard and AP Projection.}
The contract guard checks the frame, local evidence, and historical source to
determine whether a raw fact may enter $R_t$. For the representative
detailed-assistance fact shown in the figure, all of the following conditions
must hold: the frame is \texttt{affirmed} and \texttt{operational}; the
request is either direct or bound to an active source from the same risk
family; the local window contains both domain and procedural evidence; it
contains parameter/code specificity evidence; and no effective refusal is
present. These conditions exclude negated expressions, safety-oriented
analysis, isolated keywords, stale historical objectives, and replies lacking
substantive detail. Other raw facts use their corresponding fixed guard
conditions; the complete condition inventory is given in
Appendix~\ref{app:ap-contract}.

Raw facts admitted by the guard form the current-round set $R_t$. The fixed AP
contract then activates canonical APs from predefined raw-fact combinations:
\begin{ltlformula}
\begin{aligned}
a\in A_t
\quad\Longleftrightarrow\quad
R_t\text{ contains the fixed combination}\\
\text{required to activate }a.
\end{aligned}
\end{ltlformula}

For example, a user-objective evidence combination may activate
\texttt{u\_harmful\_objective}; a bound relation together with an active
source may activate \texttt{u\_continuation}; and a combination of
assistant-side procedural and specificity evidence may activate
\texttt{a\_procedural\_detail}. Refusal evidence activates
\texttt{a\_effective\_refusal} under its corresponding conditions. The complete
raw-to-final AP activation mapping is omitted from the main text and provided
in Appendix~\ref{app:ap-contract}.

The resulting $A_t$ is appended as the finite-trace position for round $t$.
Provenance from each final AP to its supporting raw facts, frames, cues, and
original text spans is retained for subsequent specification monitoring and
witness replay.

\smallskip
\noindent\textbf{Causal transaction trace.}
For turn \(t\), user-side facts are instantiated from the pre-reply prefix
\(P_t^{-}\), whereas reply-side and dialogue-side facts may use the post-reply
prefix \(P_t^{+}\); the transaction \(A_t\) and running trace $w_t=A_0A_1\cdots A_t$ are exactly as defined in Sec.~\ref{subsec:problem-formulation}.

Multiple propositions supported by the same observation coexist in $A_t$;
they do not create additional trace positions or alter the causal ordering of
the dialogue. Each position represents one complete user--assistant transaction. Speaker
types distinguish the source of individual facts within \(A_t\); consequently,
\(\mathbf{X}\) always refers to the next dialogue turn.

Trace construction is causally immutable. Once \(A_t\) has been materialized,
later messages may append new valuations but cannot revise earlier ones.
Dialogue-side APs that summarize history are activated only at the current
position and are never propagated backward. Deterministic activation therefore gives prefix
consistency (Proposition~1, Appendix~\ref{app:monitoring-proofs}). At this point, natural language has been reduced to the input word consumed by
the monitor; the \(\mathrm{LTL}_f\) specifications do not participate in AP
extraction.

\smallskip
\noindent\textbf{Compilation to finite automata.}
For each typed specification \(\phi\in\Phi\), the compiler constructs a
symbolic deterministic finite automaton
\vspace{-4pt}
\[
\mathcal{M}_{\phi}
=
\operatorname{compile}(\phi)
=
\bigl(
Q_{\phi},
2^{\mathcal{A}},
\delta_{\phi},
q_{0,\phi},
F_{\phi}
\bigr)
\]
(cf.~Definition~\ref{def:dfa}). Here, \(\delta_{\phi}\) is the transition
function of \(\mathcal{M}_{\phi}\): if the monitor is in state \(q\) and
reads the transaction-level AP valuation \(A\subseteq\mathcal{A}\), then
\(\delta_{\phi}(q,A)=q'\) means that its next state is \(q'\). For
readability, we omit the subscript \(\phi\) from the automaton components
below. We write \(A\models\chi\) when \(A\) satisfies a propositional
condition \(\chi\).

Transitions are represented symbolically by predicates over AP valuations;
the compiler therefore does not enumerate the alphabet
\(2^{\mathcal{A}}\). Compilation occurs once when the specification library
is loaded. The resulting automaton and its specification fingerprint are
then fixed.

The released specifications belong to a bounded
\(\{\mathbf{F},\mathbf{X}\}\) fragment and take one of two forms. A same-turn
specification has the form
\[
\phi=\mathbf{F}(B),
\]
where \(B\) is a propositional condition over one transaction. Its automaton
has two states,
\vspace{-4pt}
\[
Q=\{q_0,q_V\},
\qquad
F=\{q_V\},
\]
where \(q_0\) searches for a valuation satisfying \(B\), and \(q_V\) records
that such a valuation has been observed. The transition from \(q_0\) is
\vspace{-4pt}
\[
\delta(q_0,A)
=
\begin{cases}
q_V, & A\models B,\\
q_0, & \text{otherwise}.
\end{cases}
\]

An adjacent-turn specification has the form
\vspace{-4pt}
\[
\phi
=
\mathbf{F}\!\left(H\land\mathbf{X}G\right),
\]
where \(H\) is an anchor condition and \(G\) must hold in the immediately
following transaction. Its automaton has three states,
\vspace{-4pt}
\[
Q=\{q_0,q_P,q_V\},
\qquad
F=\{q_V\}.
\]
Here, \(q_0\) indicates that no unresolved anchor is active, \(q_P\) records
that the preceding transaction satisfied \(H\), and \(q_V\) records a
completed adjacent-turn match. The transitions from \(q_0\) and \(q_P\) are
\vspace{-4pt}
\[
\delta(q_0,A)
=
\begin{cases}
q_P, & A\models H,\\
q_0, & \text{otherwise},
\end{cases}
\]
and
\vspace{-4pt}
\[
\delta(q_P,A)
=
\begin{cases}
q_V, & A\models G,\\
q_P, & A\not\models G\ \land\ A\models H,\\
q_0, & \text{otherwise}.
\end{cases}
\]
The self-transition on \(q_P\) refreshes the pending anchor: if the previous
anchor does not complete a match but the current transaction again satisfies
\(H\), the current transaction becomes the new anchor.

For both forms, the accepting state is absorbing:
\vspace{-4pt}
\[
\delta(q_V,A)=q_V
\qquad
\forall A\subseteq\mathcal{A}.
\]
This is sound because appending dialogue turns cannot invalidate an already
completed existential witness. Thus, a same-turn specification requires two
states, whereas an adjacent-turn specification requires at most three.

\smallskip
\noindent\textbf{Finite-state monitoring and evidence.}
For each \(\phi\in\Phi\) and trace \(w_t=A_0\cdots A_t\), the monitor
runs the corresponding automaton \(\mathcal{M}_{\phi}\). Let
\(s_0^{\phi}=q_0^{\phi}\) and update
\vspace{-4pt}
\[
s_{j+1}^{\phi}
=
\delta_{\phi}(s_j^{\phi},A_j),
\qquad
0\leq j\leq t.
\]
The specifications matched by the current prefix are
\[
\Phi^\star(w_t)
=
\left\{
\phi\in\Phi
\ \middle|\
w_t\in L(\mathcal{M}_{\phi})
\right\}.
\]
The monitor reports \textsc{MonitoredViolation} exactly when
\(\Phi^\star(w_t)\neq\varnothing\); otherwise, the observed prefix has not
instantiated any defined forbidden pattern. This procedure performs
finite-word acceptance over the observed trace, rather than language
inclusion or model checking over all possible system behaviors.

Every automaton transition retains the event label, speaker type, AP
predicate, and relative offset inherited from its source specification.
For a same-turn specification, the transition into \(q_V\) identifies the
triggering turn. For an adjacent-turn specification, \(q_P\) stores the
anchor turn and its AP provenance, which are combined with the evidence from
the accepting transition. Each matched specification therefore produces a
structured witness
\vspace{-4pt}
\[
\omega_{\phi}
=
\left(
\phi,
\left(
    t_k,\mathcal{P}_k,\mathcal{E}_k
\right)_{k=1}^{n_{\phi}},
\operatorname{fp}(\phi),
\operatorname{dig}(w_t)
\right).
\]
where \(n_{\phi}\) is the number of matched events, \(t_k\) is a matched
turn, \(\mathcal{P}_k\) is the set of canonical APs realizing the
corresponding event, and \(\mathcal{E}_k\) contains their source evidence.
Here, \(\operatorname{fp}(\phi)\) denotes the fingerprint identifying the
frozen specification \(\phi\), and \(\operatorname{dig}(w_t)\) denotes the
digest of the evaluated trace.

Before returning a verdict, the verifier replays the match against the frozen
specification, trace digest, and raw provenance. A witness that cannot be
reproduced is rejected and cannot support a valid verdict.

\smallskip
\noindent\textbf{Runtime mode.}
When user message \(u_t\) arrives, SpecGuard first activates the
user-side APs from \(P_t^{-}\). The target LLM then produces a candidate reply
\(\hat r_t\), which remains in a private buffer. Reply-side and dialogue-side
activation yields a provisional transaction \(\hat A_t\), and every automaton
evaluates its transition on a provisional copy of the monitoring state. If no
automaton enters an accepting state, \(\hat r_t\), \(\hat A_t\), and the
updated automaton states are committed, and the reply is released. Otherwise,
the candidate is withheld, the committed trace remains unchanged, and the
verdict, matched specifications, and witnesses are passed to the deployment
policy for blocking, replacement, or human review. Excluding AP activation,
each new transaction performs one constant-width transition per specification;
the finite-state update therefore takes \(O(|\Phi|)\) time and
\(O(|\Phi|)\) state, including bounded pending-witness metadata.

\smallskip
\noindent\textbf{Audit mode.}
Audit mode receives a completed dialogue and replays its transactions in their
original order. At each turn, it reconstructs the causal AP valuation and
advances the same automata used at runtime. The audit records the first
accepting turn for each specification and returns all matched specifications
and their witnesses. For a dialogue with \(T+1\) turns, the finite-state
monitoring cost is \(O(T|\Phi|)\), excluding AP activation. Because replay is
strictly prefix ordered, future messages cannot reinterpret or modify an
earlier valuation.

\noindent\textbf{Formal robustness.}
The finite automata are executable representations of the typed
\(\mathrm{LTL}_f\) specifications, rather than an independent set of
detection rules.

\begin{theorem}[Correctness of Finite-State Monitoring]
\label{thm:finite-monitor-correctness}
For any typed specification \(\phi\in\Phi\), let
\(\mathcal{M}_{\phi}=\operatorname{compile}(\phi)\) be its corresponding
automaton. Then, for any finite AP trace \(w\),
\vspace{-4pt}
\[
w\models_f\phi
\quad\Longleftrightarrow\quad
w\in L(\mathcal{M}_{\phi}).
\]
\end{theorem}
Consequently, automaton acceptance has exactly the semantics of the source
specification. The construction further guarantees prefix-consistent traces,
equivalent runtime and audit decisions on the same prefix, invariance to
constraint-irrelevant APs and excluded endpoint facts, replayable verdict
witnesses, and persistence of accepted violations under trace extension.
Appendix~\ref{app:monitoring-proofs}
states and proves these properties. They hold at the canonical-trace boundary
and do not assume error-free natural-language AP activation, which is evaluated
separately.

\section{Evaluation}
\subsection{Experimental Setup}
\label{subsec:experimental-setup}

\smallskip
\noindent\textbf{Datasets.}
We evaluate SpecGuard on three multi-turn safety benchmarks spanning
two languages, 10,075 dialogues, and 48,933 user--assistant turns
(Table~\ref{tab:datasets}).  SafeDialBench-ZH and SafeDialBench-EN define a
positive dialogue by an assistant-side safety violation.  SafeMTData instead
labels complete ActorAttack trajectories, including attacks whose final
assistant response is a refusal.  We retain each benchmark's native
conversation-level labels and report its metrics separately; no labels,
thresholds, or scores are pooled across datasets.  Across all three benchmarks,
SpecGuard uses the same 15-proposition contract, ten typed temporal
constraints, and deterministic verdict rule.

\begin{table*}[t]
\caption{Multi-turn datasets and their evaluation semantics.}
\label{tab:datasets}
\centering
\small
\setlength{\tabcolsep}{5pt}
\renewcommand{\arraystretch}{1.10}
\begin{tabular*}{\textwidth}{@{\extracolsep{\fill}}lcccccl@{}}
\toprule
\textbf{Dataset} & \textbf{Lang.} & \textbf{Unit}
& \textbf{Dialogs} & \textbf{Turns} & \textbf{Pos./Neg.}
& \textbf{Positive-label semantics} \\
\midrule
SafeDialBench-ZH & ZH & Multi-turn & 3,602 & 16,628
& 1,981 / 1,621 & Response-side violation \\
SafeDialBench-EN & EN & Multi-turn & 3,639 & 17,913
& 2,005 / 1,634 & Response-side violation \\
SafeMTData & EN & Multi-turn & 2,834 & 14,392
& 1,680 / 1,154 & ActorAttack trajectory \\
\bottomrule
\end{tabular*}

\end{table*}

\smallskip
\noindent\textbf{Baselines and identical-input protocol.}
We compare against four representative open-weight guard models:
ShieldGemma-2B \cite{zeng2024shieldgemmagenerativeaicontent}, WildGuard-7B \cite{han2024wildguard}, BingoGuard-8B \cite{BingoGuard}, and GuardReasoner-8B~\cite{GuardReasoner}.  Every method
receives the exact same ordered user--assistant exchanges.  Following the
baselines' stateless interfaces, we classify each turn independently and emit a
dialogue-level alarm if \emph{any} turn is flagged.  SpecGuard
evaluates the same ordered exchanges as a typed finite trace and emits one
conversation-level verdict.  We use pinned official checkpoints and their
native inference templates, disable sampling, and apply no retraining,
fine-tuning, or benchmark-specific calibration.

\smallskip
\noindent\textbf{Metrics and statistical analysis.}
The evaluation unit is a completed dialogue. We use F1 as the primary
detection metric, FPR to quantify unnecessary alarms on negative dialogues,
and recall to measure positive coverage. We further report balanced accuracy
(BA),
\begin{ltlformula}
\begin{aligned}
\mathrm{BA}=\frac{\mathrm{Recall}+(1-\mathrm{FPR})}{2},
\end{aligned}
\end{ltlformula}
to summarize recall and false-alarm control at the frozen operating point.
We additionally report area under the ROC curve (AUC) to measure
discrimination across the score range.  Because the monitor's output is Boolean, its AUC is computed from a
deterministic, label-free \emph{specification-completion score}: for each legal
temporal alignment, we average the satisfied clauses, then take the maximum
over alignments and specifications.  The score uses neither labels nor learned
parameters and equals one if and only if the formal verdict is violated.
Predictions are aligned by conversation identifier, and every headline number
is computed on the complete benchmark population.  Results transcribed from
prior work are reported separately as contextual operating points and are not
treated as identical-input comparisons.

\smallskip
\noindent\textbf{Reproducibility.}
All reported runs use one frozen AP contract, activator configuration,
constraint library, and monitoring configuration.  Once fixed, no rule,
threshold, or dataset-specific decision setting is changed across benchmarks;
SpecGuard contains no learned detector weights.  We retain
per-dialogue predictions and replayable witnesses, and record cryptographic
hashes for evaluation inputs, source files, conversation identifiers, and
baseline checkpoint revisions.

\subsection{Detection Effectiveness}
\label{subsec:detection-effectiveness}

We first measure conversation-level detection under the identical-input
protocol defined above.  Figure~\ref{fig:f1-fpr-tradeoff} exposes the resulting
F1--FPR operating points, while Table~\ref{tab:main-multiturn-results} reports
the exact AUC, F1, FPR, and recall.

\begin{table*}[t]
\caption{Conversation-level performance on multi-turn jailbreak detection.
SafeDialBench is reported separately for Chinese (ZH) and English (EN).
Best values are bold.}
\label{tab:main-multiturn-results}
\centering
\small
\setlength{\tabcolsep}{1.5pt}
\renewcommand{\arraystretch}{1.08}
\begin{tabular*}{\textwidth}{
@{\extracolsep{\fill}}l
ccccccccccccccc@{}}
\toprule
& \multicolumn{5}{c}{\textbf{SafeDialBench-ZH}}
& \multicolumn{5}{c}{\textbf{SafeDialBench-EN}}
& \multicolumn{5}{c}{\textbf{SafeMTData}} \\
\cmidrule(lr){2-6}
\cmidrule(lr){7-11}
\cmidrule(lr){12-16}
\textbf{Method}
& \textbf{AUC}$\uparrow$
& \textbf{F1}$\uparrow$
& \textbf{FPR}$\downarrow$
& \textbf{Rec.}$\uparrow$
& \textbf{BA}$\uparrow$
& \textbf{AUC}$\uparrow$
& \textbf{F1}$\uparrow$
& \textbf{FPR}$\downarrow$
& \textbf{Rec.}$\uparrow$
& \textbf{BA}$\uparrow$
& \textbf{AUC}$\uparrow$
& \textbf{F1}$\uparrow$
& \textbf{FPR}$\downarrow$
& \textbf{Rec.}$\uparrow$
& \textbf{BA}$\uparrow$ \\
\midrule

WildGuard-7B
& 0.865 & 0.449 & 0.012 & 0.292 & 0.640
& 0.831 & 0.652 & 0.009 & 0.488 & 0.740
& 0.444 & 0.694 & 0.022 & 0.539 & 0.759 \\

GuardReasoner-8B
& 0.851 & 0.759 & 0.030 & 0.627 & 0.799
& 0.730 & 0.638 & 0.018 & 0.475 & 0.729
& 0.342 & 0.717 & 0.065 & 0.583 & 0.759 \\

BingoGuard-8B
& 0.712 & 0.363 & 0.012 & 0.224 & 0.606
& 0.636 & 0.271 & 0.018 & 0.159 & 0.571
& 0.445 & 0.101 & 0.076 & 0.056 & 0.490 \\

ShieldGemma-2B
& 0.868 & 0.160 & \textbf{0.007} & 0.087 & 0.540
& 0.901 & 0.173 & \textbf{0.002} & 0.095 & 0.547
& 0.816 & 0.235 & \textbf{0.017} & 0.135 & 0.559 \\

\midrule

\textbf{SpecGuard (Ours)}
& \textbf{0.984}
& \textbf{0.981}
& 0.028
& \textbf{0.984}
& \textbf{0.978}
& \textbf{0.943}
& \textbf{0.942}
& 0.107
& \textbf{0.969}
& \textbf{0.931}
& \textbf{0.967}
& \textbf{0.939}
& 0.036
& \textbf{0.907}
& \textbf{0.936} \\

\bottomrule
\end{tabular*}
\end{table*}

\begin{figure*}[t]
  \centering
  \includegraphics[width=\textwidth]{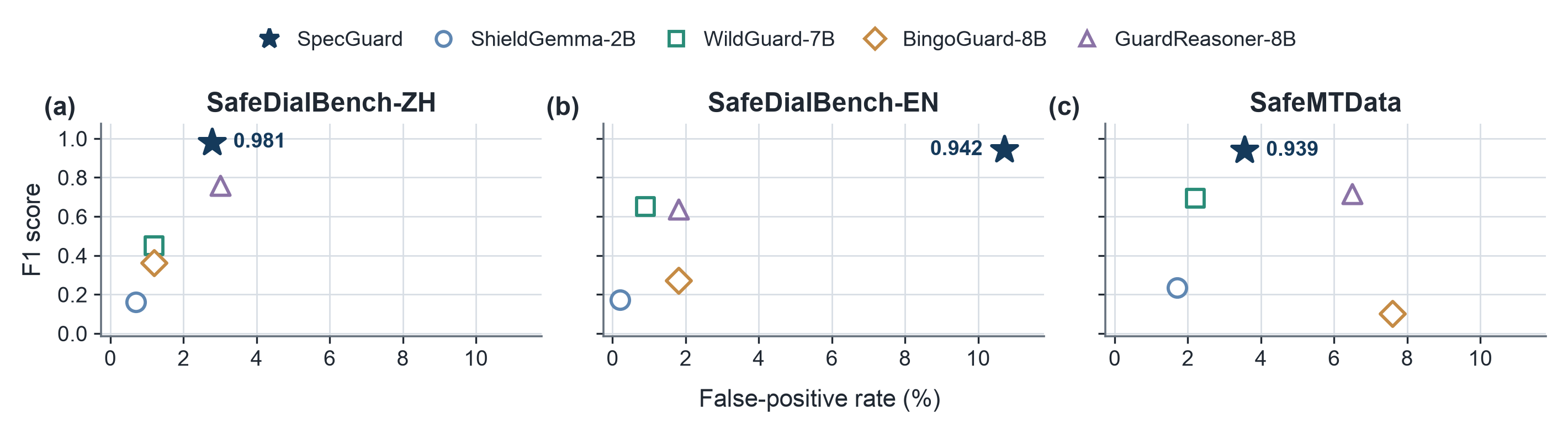}
  \caption{Conversation-level F1--FPR trade-off on the three multi-turn
  benchmarks.  Upper-left is preferable.  All methods receive identical
  conversations; external guards independently classify each turn and their
  alarms are aggregated at the conversation level.}
  \label{fig:f1-fpr-tradeoff}
\end{figure*}

\noindent\textbf{Headline performance.}
SpecGuard achieves the highest AUC, F1, and recall on all three
benchmarks.  Relative to the strongest F1 baseline
on each dataset, it improves F1 by 0.222 on SafeDialBench-ZH, 0.290 on SafeDialBench-EN, and 0.222 on SafeMTData. The agreement between threshold-free AUC and the
fixed-verdict F1 results shows that the advantage is not tied to a single
reported operating point.

Although SpecGuard does not attain the lowest FPR at its frozen operating
point, it achieves the highest BA on all three benchmarks: 0.978 on
SafeDialBench-ZH, 0.931 on SafeDialBench-EN, and 0.936 on SafeMTData. These
results demonstrate the strongest overall balance between positive-dialogue
recall and negative-dialogue specificity. In contrast, the lowest-FPR
baselines detect only 8.7--13.5\% of positive dialogues, whereas SpecGuard
recovers 90.7--98.4\% while also attaining the highest F1.

Figure~\ref{fig:f1-fpr-tradeoff} visualizes the same operating-point trade-off.
SpecGuard maintains low FPR on SafeDialBench-ZH and SafeMTData while
substantially improving both F1 and BA. SafeDialBench-EN represents its most
recall-oriented operating point, reaching 96.9\% recall at a 10.7\% FPR.
Even at this point, SpecGuard retains the highest BA of 0.931, while its
leading AUC of 0.943 further indicates that the higher FPR reflects the frozen
operating-point choice rather than weaker separation across the score range.
The Chinese and English evaluations use the same 15-proposition contract,
ten typed specifications, and frozen decision configuration, without
benchmark-specific threshold tuning.

SafeMTData deliberately tests a different endpoint: its positives are complete
ActorAttack trajectories rather than only conversations ending in a violating
assistant response.  Under this trajectory-level objective,
SpecGuard reaches 0.967 AUC, 0.939 F1, and 90.7\% recall without
changing the AP contract, constraint library, or decision rule.  This result
complements the response-side SafeDialBench evaluation by testing whether the
same monitor can recognize adversarial structure at the dialogue level.

\noindent\textbf{Contextual comparison with external defenses.}
Table~\ref{tab:reference} situates SpecGuard relative to learned
guardrails, alignment-based defenses, and a stateful multi-turn detector.  To
keep the comparison auditable, external rows reproduce the values reported by
their respective papers, while the SpecGuard row uses our results
above.  These values are presented as operating-point context rather than
identical-input measurements and are therefore kept separate from the
controlled comparison.

\begin{table}[t]
\centering
\caption{Contextual comparison with reference methods. SpecGuard
scores on SafeDialBench are averaged over ZH and EN. Best displayed results
are in \textbf{bold}.}
\label{tab:reference}
\small
\setlength{\tabcolsep}{4pt}
\renewcommand{\arraystretch}{1.08}
\resizebox{\columnwidth}{!}{%
\begin{tabular}{lcccccccc}
\toprule
& \multicolumn{4}{c}{\textbf{SafeDialBench}}
& \multicolumn{4}{c}{\textbf{SafeMTData}} \\
\cmidrule(lr){2-5}\cmidrule(lr){6-9}
\textbf{Method}
& \textbf{AUC}$\uparrow$
& \textbf{F1}$\uparrow$
& \textbf{FPR}$\downarrow$
& \textbf{Rec.}$\uparrow$
& \textbf{AUC}$\uparrow$
& \textbf{F1}$\uparrow$
& \textbf{FPR}$\downarrow$
& \textbf{Rec.}$\uparrow$ \\
\midrule
\multicolumn{9}{l}{\textit{Guardrail classifiers}} \\
R2-Guard~\cite{r2}
& 0.691 & 0.698 & 0.169 & 0.607
& 0.675 & 0.591 & 0.137 & 0.543 \\
SelfDefend~\cite{SELFDEFEND}
& 0.734 & 0.721 & 0.115 & 0.682
& 0.718 & 0.685 & 0.238 & 0.591 \\
\midrule
\multicolumn{9}{l}{\textit{Alignment-based defenses (modify model weights)}} \\
Circuit Breaker~\cite{CircuitBreaker}
& 0.831 & 0.823 & 0.071 & 0.802
& 0.824 & 0.811 & 0.084 & 0.797 \\
X-Boundary~\cite{XBoundary}
& 0.858 & 0.839 & 0.094 & 0.830
& 0.836 & 0.820 & 0.091 & 0.813 \\
TARS~\cite{TARS}
& 0.865 & 0.847 & \textbf{0.057} & 0.819
& 0.847 & 0.832 & 0.042 & 0.814 \\
\midrule
\multicolumn{9}{l}{\textit{Multi-turn stateful detectors}} \\
SafeDream~\cite{safedream}
& 0.912 & 0.901 & 0.066 & 0.918
& 0.915 & 0.904 & 0.080 & 0.897 \\
\midrule

\textbf{SpecGuard (Ours)}
& \textbf{0.964} & \textbf{0.962} & 0.067 & \textbf{0.976}
& \textbf{0.967} & \textbf{0.939} & \textbf{0.036} & \textbf{0.907} \\
\bottomrule
\end{tabular}%
}

\end{table}

\begin{table}[h]
\caption{Integrity checks of the symbolic decision path.}
\label{tab:faithfulness}
\centering
\footnotesize
\setlength{\tabcolsep}{3pt}
\renewcommand{\arraystretch}{1.05}
\begin{tabularx}{\columnwidth}{@{}Xcc@{}}
\toprule
\textbf{Check} & \textbf{Scope} & \textbf{Result} \\
\midrule
Score--verdict equivalence
& 10,075 dialogues & 0 mismatches \\
Witness replay
& 5K traces/995 alarms & 1,652/1,652 \\
Endpoint invariance
& 5K traces & 0 changes \\
Endpoint-free alphabet
& 15 APs/10 specs. & 0 prohibited \\
Regression tests
& 7 tests & 7/7 passed \\
\bottomrule
\end{tabularx}
\end{table}

SpecGuard achieves the strongest displayed AUC, F1, and recall on
SafeDialBench, with an FPR comparable to SafeDream. On SafeMTData, it also
improves both F1 and FPR over the reported reference operating points. Because
the reference results use their original evaluation protocols, these
comparisons provide context rather than replace the controlled evaluation
above.

\noindent\textbf{Where does the gain come from?}
To isolate the contribution of temporal reasoning from proposition extraction,
we evaluate three rule masks over \emph{identical} AP traces: the four
zero-offset constraints (Same-turn), the six cross-turn constraints (Temporal),
and their union (Full).  Figure~\ref{fig:temporal-contribution} reports both
the resulting F1 and the alarms contributed uniquely by the cross-turn family.

\begin{figure*}[t]
  \centering
  \includegraphics[width=\textwidth]{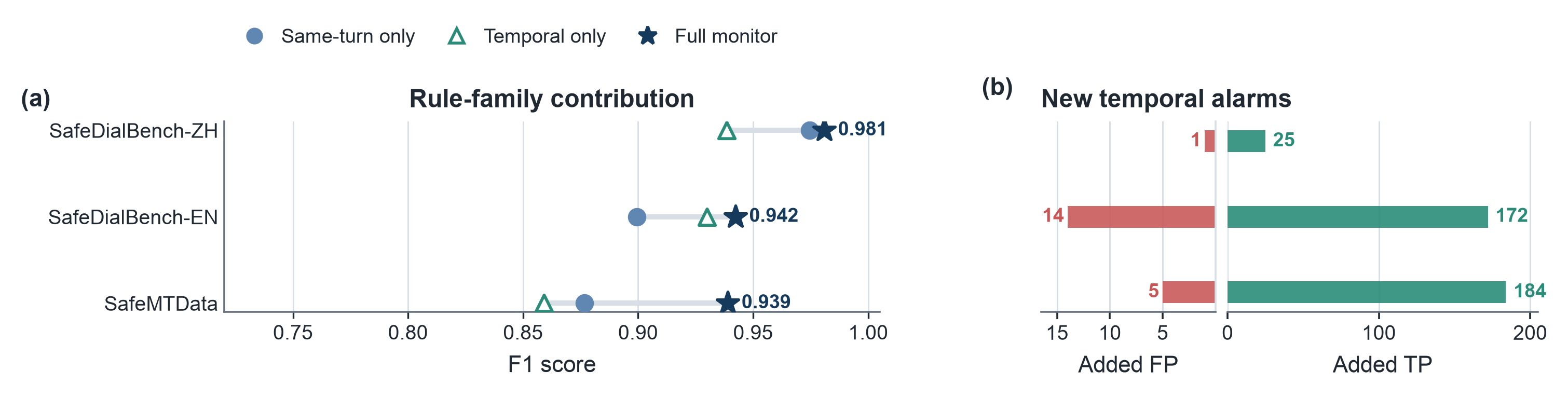}
  \caption{Contribution of cross-turn specifications under fixed AP traces.
(a) F1 of same-turn, temporal-only, and full monitors. (b) Additional true
and false positives from enabling cross-turn specifications.}
  \vspace{-6pt}
  \label{fig:temporal-contribution}
\end{figure*}

The two rule families are complementary. Adding cross-turn specifications
raises F1 from 0.975 to 0.981 on SafeDialBench-ZH, from 0.899 to 0.942 on
SafeDialBench-EN, and from 0.877 to 0.939 on SafeMTData. These gains recover
25, 172, and 184 additional true positives, with only 1, 14, and 5 additional
false positives, respectively. Temporal-only specifications also remain
effective, reaching 0.939, 0.930, and 0.859 F1. Because the AP traces are
fixed, the gains isolate temporal composition from proposition extraction.

\noindent\textbf{When does the monitor trigger?}
For each detected positive dialogue, we extract the completion turn of its
earliest canonical witness.  We normalize this location as
$p=(t_{\mathrm{trigger}}+1)/T$, where $T$ is the number of dialogue turns.
Figure~\ref{fig:first-trigger-dynamics} plots the cumulative fraction of
\emph{all} positive dialogues detected by progress $p$; consequently, each
curve's endpoint equals its final recall.

\begin{figure*}[t]
  \centering
  \includegraphics[width=\textwidth]{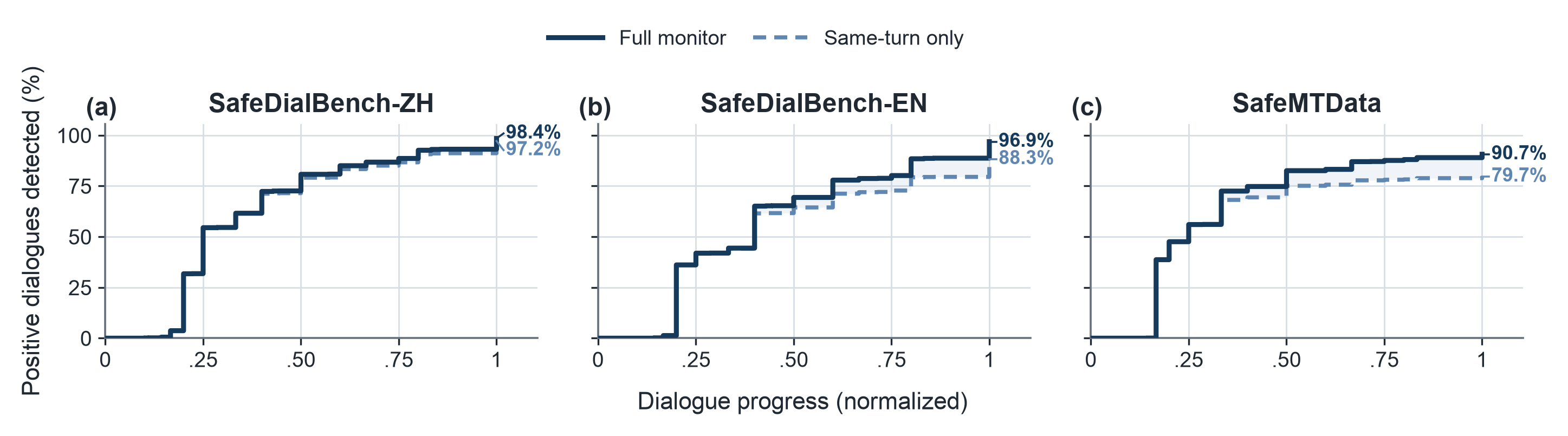}
  \caption{First-trigger dynamics over normalized dialogue progress. A trigger
marks the earliest replayable witness. Solid lines use all ten specifications,
while dashed lines use only the four same-turn specifications. The vertical
axis is normalized by all positive dialogues, so each endpoint equals recall.}
  \label{fig:first-trigger-dynamics}
\end{figure*}

The full monitor detects more positive dialogues by the midpoint and retains
this advantage through the final turn across all three datasets. Temporal
evidence therefore improves both early detection and final recall, rather than
contributing only at the last exchange.

\noindent\textbf{Stability across dialogue length.}
Table~\ref{tab:length-stratified-results} stratifies the same frozen predictions
by the number of turns.  F1 does not degrade in the longest ($\geq5$-turn)
stratum on any benchmark.  This
stratification is not used for model or threshold selection; it serves as a
stress test of trace accumulation.

\begin{table}[t]
\caption{Detection performance stratified by dialogue length. Dialogues with fewer than three turns (SafeMTData, n=152) are omitted.}
\label{tab:length-stratified-results}
\centering
\small
\setlength{\tabcolsep}{5pt}
\renewcommand{\arraystretch}{1.08}
\begin{tabular}{lccc}
\toprule
\textbf{Dataset} & \textbf{Turns} & $\boldsymbol{n}$ & \textbf{F1} \\
\midrule
SafeDialBench-ZH & 3--4 & 1,480 & 0.976 \\
SafeDialBench-ZH & $\geq5$ & 2,122 & 0.984 \\
\midrule
SafeDialBench-EN & 3--4 & 481 & 0.922 \\
SafeDialBench-EN & $\geq5$ & 3,158 & 0.945 \\
\midrule
SafeMTData & 3--4 & 659 & 0.934 \\
SafeMTData & $\geq5$ & 2,023 & 0.953 \\
\bottomrule
\end{tabular}
\end{table}

\begin{table}[t]
\caption{End-to-end CPU runtime on SafeDialBench-EN with four workers.
Per-dialogue and per-turn costs are amortized wall-clock values.}
\label{tab:runtime}
\centering
\small
\setlength{\tabcolsep}{7pt}
\renewcommand{\arraystretch}{1.10}
\begin{tabular}{lr}
\toprule
\textbf{Measurement} & \textbf{Result} \\
\midrule
Dialogues & 3,639 \\
Dialogue turns & 17,913 \\
Total wall-clock time & 246.8\,s \\
Dialogue throughput & 14.7 dialogues/s \\
Turn throughput & 72.6 turns/s \\
Amortized time per dialogue & 67.8\,ms \\
Amortized time per turn & 13.8\,ms \\
\bottomrule
\end{tabular}
\end{table}

\subsection{Faithfulness and Auditability}
\label{subsec:faithfulness}

We evaluate faithfulness against the monitor's actual decision function: given
the extracted AP trace, does the exported evidence reproduce the exact alarm?
This criterion is stronger than post-hoc plausibility.  A SpecGuard
alarm is emitted together with the violated constraint, its fingerprint, the
ordered participating turns, the matched APs and raw provenance, and a digest
of the evaluated trace.  Evidence generation is therefore part of evaluation,
not a separate explanation model.

\noindent\textbf{Decision-path replay.}
The randomized integrity audit evaluates 5,000 traces of length one to eight
drawn over the complete 15-AP alphabet.  For every emitted witness, the checker
re-evaluates the typed constraint, verifies the trace digest, and checks each
matched AP against its raw producer provenance.  All 1,652 witnesses from 995
alarmed traces replay successfully.  Separately, full-population evaluation
confirms on every benchmark dialogue that a completion score of one is exactly
equivalent to the Boolean monitor verdict.  The exported evidence therefore
identifies the same ordered observations that caused the decision rather than
serving as a post-hoc rationale.

\noindent\textbf{No verdict shortcut.}
The canonical alphabet contains 15 low-level context, user, and assistant
facts, but no proposition denoting an attack, harmful delivery, unsafe response,
or final violation.  We additionally inject every prohibited endpoint AP into
each randomized trace before re-evaluation.  In all cases, canonical
projection removes the injected APs and leaves both the decision and trace
digest unchanged.  Endpoint semantics therefore arise only from typed
constraint evaluation, while the underlying observations remain independently
inspectable.

Taken together, these checks establish a closed audit path from raw evidence to
canonical APs, from APs to an ordered witness, and from that witness to the
reported verdict.  The end-to-end effectiveness results include the extraction
stage, whereas Table~\ref{tab:faithfulness} establishes exact faithfulness of
the symbolic decision and its exported evidence. Appendix~\ref{app:witness-replay} formalizes
witness soundness and trace-bound replay.

\subsection{Runtime Performance}
\label{subsec:runtime-performance}

We measure the same end-to-end path used for the effectiveness results: natural
language AP extraction, canonical trace construction, evaluation of all ten
constraints, verdict and witness construction, and specification-completion
scoring.  SafeDialBench-EN is the largest benchmark by number of turns.  Each
of its 3,639 dialogues is processed exactly once on an Apple M4 CPU using four
worker processes; no GPU or external model service is involved.

Table~\ref{tab:runtime} shows that the complete transcript-to-score-and-verdict
path processes the 17,913-turn benchmark in 246.8 seconds.  This corresponds
to 72.6 turns/s and an amortized 13.8\,ms per turn, including evidence
generation and the label-free completion score used for AUC.  The result
therefore measures deployment-facing work rather than only formula evaluation
or a cached AP trace.

The batch evaluator constructs one canonical trace per dialogue and evaluates
the fixed library $\Phi$ of ten bounded temporal constraints.  All active
constraints have offsets of at most one round, so verdict evaluation and
completion scoring require $O(T|\Phi|)$ time and the auditable trace uses
$O(T)$ storage for a $T$-turn dialogue.  Because $|\Phi|=10$ is fixed, both
costs grow linearly with dialogue length.  No model decoding or search over
unbounded attack trajectories is required.  Together with the
length-stratified results, this shows that SpecGuard maintains stable
detection and practical CPU throughput as dialogue traces grow.

\section{Conclusion}
\label{sec:conclusion}

We presented SpecGuard, which converts multi-turn dialogues into auditable AP
traces and checks them against typed \(\mathrm{LTL}_f\) specifications.
Across languages and attack trajectories, SpecGuard achieves strong detection
performance while producing replayable witnesses for its verdicts.

Two directions follow from this work. First, SpecGuard could be extended to
multimodal interactions by grounding APs in text, images, audio, and video
within a unified temporal trace. Second, building on counterexample-guided
corrective regeneration for formally verified specifications~\cite{li2026AutoSafeLTL},
SpecGuard's replayable witnesses could enable response regeneration: when a
violation is detected, the matched constraints and evidence identify the unsafe
content to revise, and the regenerated response can be rechecked before release,
providing verifiable safety safeguards for deployed LLM systems.

\section*{Ethics Considerations}
This work is defensive: it detects multi-turn jailbreaks and introduces no
new attacks or harmful content. All experiments use publicly released
datasets; the manual annotation in this work was performed by the authors
on these public corpora, and no new data was collected from users.
Although the released specifications could inform adaptive evasion,
explicit and auditable safety rules chiefly benefit defenders and can be
revised as attacks evolve.

\bibliographystyle{IEEEtran}

\bibliography{bib}

@inproceedings{russinovich2024crescendo,
author = {Russinovich, Mark and Salem, Ahmed and Eldan, Ronen},
title = {Great, now write an article about that: the crescendo multi-turn LLM jailbreak attack},
year = {2025},
isbn = {978-1-939133-52-6},
publisher = {USENIX Association},
address = {USA},
booktitle = {Proceedings of the 34th USENIX Conference on Security Symposium},
articleno = {125},
numpages = {20},
location = {Seattle, WA, USA},
series = {SEC '25}
}

@article{ren2024actorattack,
      author    = {Qibing Ren and Hao Li and Dongrui Liu and Zhanxu Xie and Xiaoya Lu and Yu Qiao and Lei Sha and Junchi Yan and Lizhuang Ma and Jing Shao},
  title     = {{LLMs} Know Their Vulnerabilities: Uncover Safety Gaps through Natural Distribution Shifts},
  journal   = {arXiv preprint arXiv:2410.10700},
  year      = {2024},
  url       = {https://arxiv.org/abs/2410.10700}
}

@article{rahman2025xteaming,
  title={X-Teaming: Multi-Turn Jailbreaks and Defenses with Adaptive Multi-Agents},
  author={Salman Rahman and Liwei Jiang and James Shiffer and Genglin Liu and Sheriff M. Issaka and Md. Rizwan Parvez and Hamid Palangi and Kai-Wei Chang and Yejin Choi and Saadia Gabriel},
  journal={ArXiv},
  year={2025},
  volume={abs/2504.13203},
  url={https://api.semanticscholar.org/CorpusID:277940089}
}

@misc{chao2023pair,
      title={Jailbreaking Black Box Large Language Models in Twenty Queries}, 
      author={Patrick Chao and Alexander Robey and Edgar Dobriban and Hamed Hassani and George J. Pappas and Eric Wong},
      year={2023},
      eprint={2310.08419},
      archivePrefix={arXiv},
      primaryClass={cs.LG}
}

@misc{inan2023llamaguard,
      title={Llama Guard: LLM-based Input-Output Safeguard for Human-AI Conversations}, 
      author={Hakan Inan and Kartikeya Upasani and Jianfeng Chi and Rashi Rungta and Krithika Iyer and Yuning Mao and Michael Tontchev and Qing Hu and Brian Fuller and Davide Testuggine and Madian Khabsa},
      year={2023},
      eprint={2312.06674},
      archivePrefix={arXiv},
      primaryClass={cs.CL},
      url={https://arxiv.org/abs/2312.06674}, 
}

@inproceedings{han2024wildguard,
author = {Han, Seungju and Rao, Kavel and Ettinger, Allyson and Jiang, Liwei and Lin, Bill Yuchen and Lambert, Nathan and Choi, Yejin and Dziri, Nouha},
title = {WILDGUARD: open one-stop moderation tools for safety risks, jailbreaks, and refusals of LLMs},
year = {2024},
isbn = {9798331314385},
publisher = {Curran Associates Inc.},
address = {Red Hook, NY, USA},
booktitle = {Proceedings of the 38th International Conference on Neural Information Processing Systems},
articleno = {261},
numpages = {39},
location = {Vancouver, BC, Canada},
series = {NIPS '24}
}

@misc{zou2023universal,
      title={Universal and Transferable Adversarial Attacks on Aligned Language Models}, 
      author={Andy Zou and Zifan Wang and Nicholas Carlini and Milad Nasr and J. Zico Kolter and Matt Fredrikson},
      year={2023},
      eprint={2307.15043},
      archivePrefix={arXiv},
      primaryClass={cs.CL},
      url={https://arxiv.org/abs/2307.15043}, 
}

@misc{wei2023jailbroken,
      title={Jailbroken: How Does LLM Safety Training Fail?}, 
      author={Alexander Wei and Nika Haghtalab and Jacob Steinhardt},
      year={2023},
      eprint={2307.02483},
      archivePrefix={arXiv},
      primaryClass={cs.LG},
      url={https://arxiv.org/abs/2307.02483}, 
}

@inproceedings{anil2024manyshot,
author = {Anil, Cem and Durmus, Esin and Panickssery, Nina and Sharma, Mrinank and Benton, Joe and Kundu, Sandipan and Batson, Joshua and Tong, Meg and Mu, Jesse and Ford, Daniel and Mosconi, Fracesco and Agrawal, Rajashree and Schaeffer, Rylan and Bashkansky, Naomi and Svenningsen, Samuel and Lambert, Mike and Radhakrishnan, Ansh and Denison, Carson and Hubinger, Evan J and Bai, Yuntao and Bricken, Trenton and Maxwell, Timothy and Schiefer, Nicholas and Sully, James and Tamkin, Alex and Lanhan, Tamera and Nguyen, Karina and Korbak, Tomasz and Kaplan, Jared and Ganguli, Deep and Bowman, Samuel R. and Perez, Ethan and Grosse, Roger Baker and Duvenaud, David},
title = {Many-shot Jailbreaking},
year = {2024},
isbn = {9798331314385},
publisher = {Curran Associates Inc.},
address = {Red Hook, NY, USA},
booktitle = {Proceedings of the 38th International Conference on Neural Information Processing Systems},
articleno = {4121},
numpages = {47},
location = {Vancouver, BC, Canada},
series = {NIPS '24}
}

@inproceedings{xu2024safedecoding,
    title = "{S}afe{D}ecoding: Defending against Jailbreak Attacks via Safety-Aware Decoding",
    author = "Xu, Zhangchen  and
      Jiang, Fengqing  and
      Niu, Luyao  and
      Jia, Jinyuan  and
      Lin, Bill Yuchen  and
      Poovendran, Radha",
    editor = "Ku, Lun-Wei  and
      Martins, Andre  and
      Srikumar, Vivek",
    booktitle = "Proceedings of the 62nd Annual Meeting of the Association for Computational Linguistics (Volume 1: Long Papers)",
    month = aug,
    year = "2024",
    address = "Bangkok, Thailand",
    publisher = "Association for Computational Linguistics",
    url = "https://aclanthology.org/2024.acl-long.303/",
    doi = "10.18653/v1/2024.acl-long.303",
    pages = "5587--5605"
}

@misc{li2026AutoSafeLTL,
      title={Automatic Generation of Safety-compliant Linear Temporal Logic via Large Language Model: A Self-supervised Framework}, 
      author={Junle Li and Siqi Chen and Jiakai Li and Meiqi Tian and Bingzhuo Zhong},
      year={2025},
      eprint={2503.15840},
      archivePrefix={arXiv},
      primaryClass={cs.LO},
      url={https://arxiv.org/abs/2503.15840}, 
}

@misc{cosler2023nl2spec,
      title={nl2spec: Interactively Translating Unstructured Natural Language to Temporal Logics with Large Language Models}, 
      author={Matthias Cosler and Christopher Hahn and Daniel Mendoza and Frederik Schmitt and Caroline Trippel},
      year={2023},
      eprint={2303.04864},
      archivePrefix={arXiv},
      primaryClass={cs.LO},
      url={https://arxiv.org/abs/2303.04864}, 
}

@misc{chen2024nl2tl,
      title={NL2TL: Transforming Natural Languages to Temporal Logics using Large Language Models}, 
      author={Yongchao Chen and Rujul Gandhi and Yang Zhang and Chuchu Fan},
      year={2024},
      eprint={2305.07766},
      archivePrefix={arXiv},
      primaryClass={cs.CL},
      url={https://arxiv.org/abs/2305.07766}, 
}

@inproceedings{chen2025ShieldAgent,
  author    = {Mintong Kang and Bo Li},
  title     = {{R$^{2}$-Guard}: Robust Reasoning Enabled {LLM} Guardrail via Knowledge-Enhanced Logical Reasoning},
  booktitle = {International Conference on Learning Representations},
  year      = {2025},
  url       = {https://openreview.net/forum?id=CkgKSqZbuC}
}

@inproceedings{Lee2025VeriPlan, series={CHI ’25},
   title={VeriPlan: Integrating Formal Verification and LLMs into End-User Planning},
   url={http://dx.doi.org/10.1145/3706598.3714113},
   DOI={10.1145/3706598.3714113},
   booktitle={Proceedings of the 2025 CHI Conference on Human Factors in Computing Systems},
   publisher={ACM},
   author={Lee, Christine P. and Porfirio, David and Wang, Xinyu Jessica and Zhao, Kevin Chenkai and Mutlu, Bilge},
   year={2025},
   month=Apr, pages={1–19},
   collection={CHI ’25} }

@article{FERRARI2025Formalrequirements,
title = {Formal requirements engineering and large language models: A two-way roadmap},
journal = {Information and Software Technology},
volume = {181},
pages = {107697},
year = {2025},
issn = {0950-5849},
doi = {https://doi.org/10.1016/j.infsof.2025.107697},
url = {https://www.sciencedirect.com/science/article/pii/S0950584925000369},
author = {Alessio Ferrari and Paola Spoletini}
}

@article{Mehdipour2023FormalMethodsDriving,
title = {Formal methods to comply with rules of the road in autonomous driving: State of the art and grand challenges},
journal = {Automatica},
volume = {152},
pages = {110692},
year = {2023},
issn = {0005-1098},
doi = {https://doi.org/10.1016/j.automatica.2022.110692},
url = {https://www.sciencedirect.com/science/article/pii/S0005109822005568},
author = {Noushin Mehdipour and Matthias Althoff and Radboud Duintjer Tebbens and Calin Belta}
}

@article{Wongpiromsarn2023FormalMethodsAutonomousSystems,
   title={Formal Methods for Autonomous Systems},
   volume={10},
   ISSN={2325-6826},
   url={http://dx.doi.org/10.1561/2600000029},
   DOI={10.1561/2600000029},
   number={3-4},
   journal={Foundations and Trends® in Systems and Control},
   publisher={Emerald},
   author={Wongpiromsarn, Tichakorn and Ghasemi, Mahsa and Cubuktepe, Murat and Bakirtzis, Georgios and Carr, Steven and Karabag, Mustafa O. and Neary, Cyrus and Gohari, Parham and Topcu, Ufuk},
   year={2023},
   month=Sept, pages={180–407} }

@inproceedings{Desai2017SafeRobots,
    author = {Ankush Desai and Tommaso Dreossi and Sanjit Seshia},
    title = {Combining Model Checking and Runtime Verification
              for Safe Robotics},
    booktitle = {The 17th International Conference on Runtime
              Verification},
    year = {2017},
    URL = {http://terraswarm.org/pubs/955.html}
}

@misc{sun2022multiagentmotionplanningsignal,
      title={Multi-agent Motion Planning from Signal Temporal Logic Specifications}, 
      author={Dawei Sun and Jingkai Chen and Sayan Mitra and Chuchu Fan},
      year={2022},
      eprint={2201.05247},
      archivePrefix={arXiv},
      primaryClass={cs.RO},
      url={https://arxiv.org/abs/2201.05247}, 
}

@article{Alshiekh2018Shielding, title={Safe Reinforcement Learning via Shielding}, 
volume={32}, 
url={https://ojs.aaai.org/index.php/AAAI/article/view/11797}, 
DOI={10.1609/aaai.v32i1.11797} , number={1}, journal={Proceedings of the AAAI Conference on Artificial Intelligence}, author={Alshiekh, Mohammed and Bloem, Roderick and Ehlers, Rüdiger and Könighofer, Bettina and Niekum, Scott and Topcu, Ufuk}, year={2018}, month={Apr.} }

@misc{Zhang2024LLMFMRoadmap,
      title={The Fusion of Large Language Models and Formal Methods for Trustworthy AI Agents: A Roadmap}, 
      author={Yedi Zhang and Yufan Cai and Xinyue Zuo and Xiaokun Luan and Kailong Wang and Zhe Hou and Yifan Zhang and Zhiyuan Wei and Meng Sun and Jun Sun and Jing Sun and Jin Song Dong},
      year={2024},
      eprint={2412.06512},
      archivePrefix={arXiv},
      primaryClass={cs.AI},
      url={https://arxiv.org/abs/2412.06512}, 
}

@misc{SafeDialBench,
      title={SafeDialBench: A Fine-Grained Safety Evaluation Benchmark for Large Language Models in Multi-Turn Dialogues with Diverse Jailbreak Attacks}, 
      author={Hongye Cao and Sijia Jing and Yanming Wang and Ziyue Peng and Zhixin Bai and Zhe Cao and Meng Fang and Fan Feng and Boyan Wang and Jiaheng Liu and Tianpei Yang and Jing Huo and Yang Gao and Fanyu Meng and Xi Yang and Chao Deng and Junlan Feng},
      year={2026},
      eprint={2502.11090},
      archivePrefix={arXiv},
      primaryClass={cs.CL},
      url={https://arxiv.org/abs/2502.11090}, 
}

@misc{SafeMTData,
      title={LLMs know their vulnerabilities: Uncover Safety Gaps through Natural Distribution Shifts}, 
      author={Qibing Ren and Hao Li and Dongrui Liu and Zhanxu Xie and Xiaoya Lu and Yu Qiao and Lei Sha and Junchi Yan and Lizhuang Ma and Jing Shao},
      year={2026},
      eprint={2410.10700},
      archivePrefix={arXiv},
      primaryClass={cs.CL},
      url={https://arxiv.org/abs/2410.10700}, 
}

@misc{safedream,
      title={SafeDream: Safety World Model for Proactive Early Jailbreak Detection}, 
      author={Bo Yan and Weikai Lin and Yada Zhu and Song Wang},
      year={2026},
      eprint={2604.16824},
      archivePrefix={arXiv},
      primaryClass={cs.CR},
      url={https://arxiv.org/abs/2604.16824}, 
}

@inproceedings{r2,
 author = {Kang, Mintong and Li, Bo},
 booktitle = {International Conference on Learning Representations},
 editor = {Y. Yue and A. Garg and N. Peng and F. Sha and R. Yu},
 pages = {63859--63876},
 title = {{R}$^2$-{Guard}: Robust Reasoning Enabled LLM Guardrail via Knowledge-Enhanced Logical Reasoning},
 volume = {2025},
 year = {2025}
}

@inproceedings{BingoGuard,
  author    = {Fan Yin and Philippe Laban and Xiangyu Peng and Yilun Zhou and Yixin Mao and Vaibhav Vats and Linnea Ross and Divyansh Agarwal and Caiming Xiong and Chien-Sheng Wu},
  title     = {{BingoGuard}: {LLM} Content Moderation Tools with Risk Levels},
  booktitle = {International Conference on Learning Representations},
  year      = {2025},
  url       = {https://openreview.net/forum?id=782b6152c04e9948c2cb3833e9a288ef}
}

@misc{GuardReasoner,
      title={GuardReasoner: Towards Reasoning-based LLM Safeguards}, 
      author={Yue Liu and Hongcheng Gao and Shengfang Zhai and Yufei He and Jun Xia and Zhengyu Hu and Yulin Chen and Xihong Yang and Jiaheng Zhang and Stan Z. Li and Hui Xiong and Bryan Hooi},
      year={2025},
      eprint={2501.18492},
      archivePrefix={arXiv},
      primaryClass={cs.CR},
      url={https://arxiv.org/abs/2501.18492}, 
}

@inproceedings{SELFDEFEND,
author = {Wang, Xunguang and Wu, Daoyuan and Ji, Zhenlan and Li, Zongjie and Ma, Pingchuan and Wang, Shuai and Li, Yingjiu and Liu, Yang and Liu, Ning and Rahmel, Juergen},
title = {SELFDEFEND: LLMs can defend themselves against jailbreaking in a practical manner},
year = {2025},
isbn = {978-1-939133-52-6},
publisher = {USENIX Association},
address = {USA},
booktitle = {Proceedings of the 34th USENIX Conference on Security Symposium},
articleno = {126},
numpages = {20},
location = {Seattle, WA, USA},
series = {SEC '25}
}

@inproceedings{CircuitBreaker,
 author = {Zou, Andy and Phan, Long and Wang, Justin and Duenas, Derek and Lin, Maxwell and Andriushchenko, Maksym and Wang, Rowan and Kolter, Zico and Fredrikson, Matt and Hendrycks, Dan},
 booktitle = {Advances in Neural Information Processing Systems},
 doi = {10.52202/079017-2651},
 editor = {A. Globerson and L. Mackey and D. Belgrave and A. Fan and U. Paquet and J. Tomczak and C. Zhang},
 pages = {83345--83373},
 publisher = {Curran Associates, Inc.},
 title = {Improving Alignment and Robustness with Circuit Breakers},
 volume = {37},
 year = {2024}
}

@inproceedings{XBoundary,
    title = "{X}-Boundary: Establishing Exact Safety Boundary to Shield {LLM}s from Jailbreak Attacks without Compromising Usability",
    author = "Lu, Xiaoya  and
      Liu, Dongrui  and
      Yu, Yi  and
      Xu, Luxin  and
      Shao, Jing",
    editor = "Christodoulopoulos, Christos  and
      Chakraborty, Tanmoy  and
      Rose, Carolyn  and
      Peng, Violet",
    booktitle = "Findings of the Association for Computational Linguistics: EMNLP 2025",
    month = nov,
    year = "2025",
    address = "Suzhou, China",
    publisher = "Association for Computational Linguistics",
    url = "https://aclanthology.org/2025.findings-emnlp.282/",
    doi = "10.18653/v1/2025.findings-emnlp.282",
    pages = "5247--5272",
    ISBN = "979-8-89176-335-7"
}

@inproceedings{TARS,
 author = {Kim, Taeyoun and Tajwar, Fahim and Raghunathan, Aditi and Kumar, Aviral},
 booktitle = {Advances in Neural Information Processing Systems},
 editor = {D. Belgrave and C. Zhang and H. Lin and R. Pascanu and P. Koniusz and M. Ghassemi and N. Chen},
 pages = {77024--77073},
 publisher = {Curran Associates, Inc.},
 title = {Reasoning as an Adaptive Defense for Safety},
 volume = {38},
 year = {2025}
}

@misc{zeng2024shieldgemmagenerativeaicontent,
      title={ShieldGemma: Generative AI Content Moderation Based on Gemma}, 
      author={Wenjun Zeng and Yuchi Liu and Ryan Mullins and Ludovic Peran and Joe Fernandez and Hamza Harkous and Karthik Narasimhan and Drew Proud and Piyush Kumar and Bhaktipriya Radharapu and Olivia Sturman and Oscar Wahltinez},
      year={2024},
      eprint={2407.21772},
      archivePrefix={arXiv},
      primaryClass={cs.CL},
      url={https://arxiv.org/abs/2407.21772}, 
}

@misc{ma2026thrd,
      title={THRD: A Training-Free Multi-Turn Defense Framework for Jailbreak Attacks on Large Language Models}, 
      author={Zhiqing Ma and Zhonghao Xu and Dong Yu and Chen Kang and Changliang Li and Pengyuan Liu},
      year={2026},
      eprint={2606.01738},
      archivePrefix={arXiv},
      primaryClass={cs.CL},
      url={https://arxiv.org/abs/2606.01738}, 
}

@misc{huang2025gguard,
      title={Attention-Aware GNN-based Input Defense against Multi-Turn LLM Jailbreak}, 
      author={Zixuan Huang and Kecheng Huang and Lihao Yin and Bowei He and Huiling Zhen and Mingxuan Yuan and Zili Shao},
      year={2025},
      eprint={2507.07146},
      archivePrefix={arXiv},
      primaryClass={cs.LG},
      url={https://arxiv.org/abs/2507.07146}, 
}

@misc{shen2026turngate,
      title={One Turn Too Late: Response-Aware Defense Against Hidden Malicious Intent in Multi-Turn Dialogue}, 
      author={Xinjie Shen and Rongzhe Wei and Peizhi Niu and Haoyu Wang and Ruihan Wu and Eli Chien and Bo Li and Pin-Yu Chen and Pan Li},
      year={2026},
      eprint={2605.05630},
      archivePrefix={arXiv},
      primaryClass={cs.CL},
      url={https://arxiv.org/abs/2605.05630}, 
}

@inproceedings{shen2024dan,
author = {Shen, Xinyue and Chen, Zeyuan and Backes, Michael and Shen, Yun and Zhang, Yang},
title ={{Do Anything Now}: Characterizing and Evaluating In-the-Wild Jailbreak Prompts on Large Language Models},
year = {2024},
isbn = {9798400706363},
publisher = {Association for Computing Machinery},
address = {New York, NY, USA},
url = {https://doi.org/10.1145/3658644.3670388},
doi = {10.1145/3658644.3670388},
booktitle = {Proceedings of the 2024 on ACM SIGSAC Conference on Computer and Communications Security},
pages = {1671–1685},
numpages = {15},
location = {Salt Lake City, UT, USA},
series = {CCS '24}
}

@inproceedings{liu2024hitchhiker,
author = {Liu, Yi and Deng, Gelei and Xu, Zhengzi and Li, Yuekang and Zheng, Yaowen and Zhang, Ying and Zhao, Lida and Zhang, Tianwei and Wang, Kailong},
title = {A Hitchhiker’s Guide to Jailbreaking ChatGPT via Prompt Engineering},
year = {2024},
isbn = {9798400706721},
publisher = {Association for Computing Machinery},
address = {New York, NY, USA},
url = {https://doi.org/10.1145/3663530.3665021},
doi = {10.1145/3663530.3665021},
booktitle = {Proceedings of the 4th International Workshop on Software Engineering and AI for Data Quality in Cyber-Physical Systems/Internet of Things},
pages = {12–21},
numpages = {10},
location = {Porto de Galinhas, Brazil},
series = {SEA4DQ 2024}
}

@misc{zeng2024airbench,
      title={AIR-Bench 2024: A Safety Benchmark Based on Risk Categories from Regulations and Policies}, 
      author={Yi Zeng and Yu Yang and Andy Zhou and Jeffrey Ziwei Tan and Yuheng Tu and Yifan Mai and Kevin Klyman and Minzhou Pan and Ruoxi Jia and Dawn Song and Percy Liang and Bo Li},
      year={2024},
      eprint={2407.17436},
      archivePrefix={arXiv},
      primaryClass={cs.CY},
      url={https://arxiv.org/abs/2407.17436}, 
}

@book{sipser2013introduction,
  author    = {Sipser, Michael},
  title     = {Introduction to the Theory of Computation},
  edition   = {3rd},
  publisher = {Cengage Learning},
  year      = {2013}
}

@misc{markov2023holistic,
      title={A Holistic Approach to Undesired Content Detection in the Real World}, 
      author={Todor Markov and Chong Zhang and Sandhini Agarwal and Tyna Eloundou and Teddy Lee and Steven Adler and Angela Jiang and Lilian Weng},
      year={2023},
      eprint={2208.03274},
      archivePrefix={arXiv},
      primaryClass={cs.CL},
      url={https://arxiv.org/abs/2208.03274}, 
}

@inproceedings{mazeika2024harmbench,
author = {Mazeika, Mantas and Phan, Long and Yin, Xuwang and Zou, Andy and Wang, Zifan and Mu, Norman and Sakhaee, Elham and Li, Nathaniel and Basart, Steven and Li, Bo and Forsyth, David and Hendrycks, Dan},
title = {HarmBench: a standardized evaluation framework for automated red teaming and robust refusal},
year = {2024},
publisher = {JMLR.org},
booktitle = {Proceedings of the 41st International Conference on Machine Learning},
articleno = {1431},
numpages = {44},
location = {Vienna, Austria},
series = {ICML'24}
}

@misc{zhao2024wildchat,
      title={WildChat: 1M ChatGPT Interaction Logs in the Wild}, 
      author={Wenting Zhao and Xiang Ren and Jack Hessel and Claire Cardie and Yejin Choi and Yuntian Deng},
      year={2024},
      eprint={2405.01470},
      archivePrefix={arXiv},
      primaryClass={cs.CL},
      url={https://arxiv.org/abs/2405.01470}, 
}



\appendix

\subsection{Evidence Grounding: Matcher and Output Contract}
\label{app:evidence-grounding}

This appendix fixes the deterministic text-grounding contract. We count one
logical semantic predicate with a stable output slot as one matcher; lexical,
regular-expression, and structural branches are implementation variants of that
predicate. The implementation has 52 text-grounding matchers in nine evidence
families (Table~\ref{tab:grounding-matcher-inventory}) and six bounded-relation matchers (Table~\ref{tab:relation-matcher-inventory}), which consume evidence records.

\smallskip
\noindent\textbf{Input scope and windows.}

For turn $t$, matchers receive $(u_t,r_t,W)$; $H_{t-1}$ is retained for
relation composition. Messages are normalized, routed to the Chinese or
English pattern set, and segmented into content units while preserving list and
line structure. Windows contain at most three neighboring units. In long user
messages, a request-anchored command zone excludes quoted/background text and
unanchored fenced code. These steps restrict matcher scope; they do not
activate APs.

\begin{table*}[t]
\centering
\scriptsize
\setlength{\tabcolsep}{4pt}
\renewcommand{\arraystretch}{1.03}
\caption{Logical text-grounding matchers by evidence family.}
\label{tab:grounding-matcher-inventory}
\begin{tabularx}{\textwidth}{@{}p{0.18\textwidth}c p{0.43\textwidth}X@{}}
\toprule
\textbf{Evidence family} & \textbf{Count} & \textbf{Logical matchers} &
\textbf{Trigger and output} \\
\midrule
Domain and context & 4 &
\texttt{HarmDomain}, \texttt{HighRiskDomain},
\texttt{ContextCorePhrase}, \texttt{DomainWindowBinding} &
Domain lexicon/phrase; emits domain and same-window capability support. \\
Procedure and action & 5 &
\texttt{StepSequence}, \texttt{NumberedInstruction},
\texttt{OperationalVerb}, \texttt{ActionVerb},
\texttt{CommandStructure} &
List/step/command structure or operational verb; emits marker/action, span,
and count. \\
Parameter and code & 4 &
\texttt{NumericUnit}, \texttt{Ratio}, \texttt{CodeFence},
\texttt{CommandParameter} &
Number--unit, ratio, code, command, or named parameter; emits value and span. \\
Request cues & 8 &
\texttt{DirectHowTo}, \texttt{ExactQuantityRequest},
\texttt{StepByStepRequest}, \texttt{MethodRequest},
\texttt{InstrumentRequest}, \texttt{PreparationRequest},
\texttt{ReversalRequest}, \texttt{SuppressionRequest} &
Interrogative/imperative pattern for method, amount, sequence, instrument,
preparation, reversal, or suppression; emits request cue. \\
Objective and family & 6 &
\texttt{HarmfulObjective}, \texttt{ObjectiveFamily},
\texttt{ActorRole}, \texttt{TargetRole}, \texttt{IdentifiedPerson},
\texttt{ObjectiveReference} &
Objective/family phrase, actor--target relation, identifier, or anaphora;
emits objective, family, target, or reference. \\
Evasion and intent & 7 &
\texttt{InstructionOverride}, \texttt{PersonaMode},
\texttt{ConstraintSuppression}, \texttt{PolarityInversion},
\texttt{IntentEscalation}, \texttt{ContinuationCue},
\texttt{EvasionCue} &
Override, role-play, suppression, reversal, escalation, continuation, or
bypass phrase; emits evasion/intent cue. \\
Polarity, scope, and speech act & 6 &
\texttt{Negation}, \texttt{RefusalCue}, \texttt{SafetyDisclaimer},
\texttt{EducationalScope}, \texttt{OperationalScope},
\texttt{SpeechAct} &
Negation/refusal, safety/educational/operational scope, or speech-act pattern;
emits polarity, scope, or act. \\
Reply delivery and content & 7 &
\texttt{SubstantiveContent}, \texttt{ProceduralContent},
\texttt{MethodContent}, \texttt{ActionEnablement},
\texttt{TargetedContent}, \texttt{PartialAssistance},
\texttt{HarmfulArtifact} &
Payload length, procedure, method, executable action, target detail, partial
completion, or harmful artifact; emits content support. \\
Refusal and assistance & 5 &
\texttt{EffectiveRefusal}, \texttt{SafetyRefusal},
\texttt{CapabilityRefusal}, \texttt{Boilerplate},
\texttt{NonBoilerplateContent} &
Safety/capability refusal, boilerplate, or retained payload; emits refusal and
content polarity. \\
\midrule
\multicolumn{2}{r}{\textbf{Total}} & \textbf{52 logical text matchers} & \\
\bottomrule
\end{tabularx}
\end{table*}

\smallskip
\noindent\textbf{Deterministic triggering.}

Matchers are deterministic predicates over a normalized unit or window:
lexical predicates use bilingual dictionaries and boundary-aware patterns;
structural predicates inspect list/line order, code delimiters, and numeric
formats; co-occurrence predicates are window-bounded. Negation, scope, and
refusal are emitted explicitly, and every hit returns its value, normalized
span, and supporting phrase or structural condition.

\smallskip
\noindent\textbf{Evidence output contract.}

For a window $W$, matcher $m$ returns zero or more provenance-bearing
intermediate records, not canonical APs or verdicts.

The matcher output is represented as
\[
\mathsf{match}_m(W)=
\{(\sigma,v,[b,e),s)\},
\]
and each tuple is lifted to an evidence record
\[
e=\langle\rho,m,\sigma,v,[b,e),s\rangle.
\]
Here $\rho$ is the role, $\sigma$ a typed slot, $v$ its value, $[b,e)$ the
normalized-text span, and $s$ the supporting phrase, predicate, or count. A
window may yield multiple records, including repeated hits; deduplication is
performed only after provenance is retained.

\begin{table}[t]
\centering
\scriptsize
\setlength{\tabcolsep}{3pt}
\renewcommand{\arraystretch}{1.00}
\caption{Bounded-relation matchers used after evidence grounding.}
\label{tab:relation-matcher-inventory}
\begin{tabularx}{\columnwidth}{@{}p{0.29\columnwidth}X@{}}
\toprule
\textbf{Relation matcher} & \textbf{Evidence-level condition} \\
\midrule
\texttt{ObjectiveSource} & Binds a request to its earliest objective evidence. \\
\texttt{ActiveSource} & Carries an objective while its source remains active. \\
\texttt{ActionAnaphora} & Resolves an action anaphor under an affirmed operational request. \\
\texttt{CanonicalReactivation} & Detects resumed objectives after continuation. \\
\texttt{ResilientLineage} & Preserves source lineage across bounded links. \\
\texttt{SameFamilyDelivery} & Binds request/reply evidence with equal risk family. \\
\bottomrule
\end{tabularx}
\end{table}

Relation matchers operate on evidence, frame cues, and retained source
identifiers rather than raw text. Capability and specificity are derived
combinations:
\begin{ltlformula}
\begin{aligned}
&\mathsf{Capability}_X(W)=
\mathsf{Domain}_X(W)\land\mathsf{ProceduralMarker}(W),
\qquad\\
&\mathsf{Specificity}_X(W)=
\mathsf{Capability}_X(W)\land\mathsf{ParameterOrCode}(W).
\end{aligned}
\end{ltlformula}
Thus, a parameter or domain hit alone is insufficient. The records then proceed
to frame construction and contract guards.

\subsection{AP Extraction and Provenance}
\label{app:ap-contract}

The main text defines the 15 canonical APs. Here we give only the projection,
endpoint exclusion, and provenance needed to reproduce the trace. User APs use
\(P_t^-\), reply APs use \(P_t^+\), and history APs summarize the prefix.

\smallskip
\noindent\textbf{Raw-to-canonical projection.}

Let \(R_t\) be the raw evidence for turn \(t\). Projection is multi-label and
also accepts a canonical identifier as normalized raw evidence.

The following shorthand states the complete mapping. A terminal \texttt{*} is
a prefix wildcard; \texttt{p4*} tokens are frozen producer identifiers.

\begin{description}
\setlength{\itemsep}{0pt}
\setlength{\parsep}{0pt}
\setlength{\topsep}{2pt}
\item[$\mathcal{D}_{c}$]
\texttt{violence}, \texttt{terrorism}, \texttt{drug\_weapon},
\texttt{financial\_crime}, \texttt{hate\_speech},
\texttt{discrimination}, \texttt{self\_harm}, \texttt{child\_abuse},
\texttt{sexual\_content}, \texttt{misinformation},
\texttt{unethical\_behavior}, and \texttt{privacy}.

\item[$\mathcal{D}_{h}$]
The elements of \(\mathcal{D}_{c}\), together with
\texttt{social\_misconduct}.

\item[$\mathcal{F}_{o}$]
\texttt{self\_harm\_method}, \texttt{suicide\_planning},
\texttt{self\_harm\_encouragement}, \texttt{credential\_intrusion},
\texttt{monitoring\_evasion}, \texttt{discriminatory\_abuse},
\texttt{privacy\_intrusion}, \texttt{illicit\_operations},
\texttt{abusive\_generation}, \texttt{harmful\_rationalization}.

\item[$\mathcal{H}$]
\texttt{u\_requests\_abusive\_language\_examples},
\texttt{u\_requests\_derogatory\_generation},
\texttt{u\_requests\_false\_claim\_generation},
\texttt{u\_requests\_false\_claim\_persuasion},
\texttt{u\_requests\_harmful\_actor\_methods},
\texttt{u\_requests\_harmful\_advocacy},
\texttt{u\_requests\_harmful\_case\_examples},
\texttt{u\_requests\_harmful\_scene},
\texttt{u\_requests\_harmful\_enumeration},
\texttt{u\_requests\_harmful\_mechanism},
\texttt{u\_requests\_illicit\_artifact\_method},
\texttt{u\_requests\_privacy\_compromise},
\texttt{u\_zh\_requests\_concrete\_harmful\_condu}.

\item[$\mathcal{P}_{h}$]
\texttt{u\_p4k\_direct\_*}, \texttt{u\_p4k\_context\_*},
\texttt{u\_p4n\_direct\_*}, \texttt{u\_p4s\_direct\_*},
\texttt{u\_p4s\_bound\_*}, \texttt{u\_p4t\_bound\_*},
\texttt{u\_p4u\_bound\_*}, \texttt{u\_p4o\_bound\_*}.

\item[$\mathcal{P}_{x}$]
\texttt{u\_p4k\_direct\_*}, \texttt{u\_p4n\_direct\_*},
\texttt{u\_p4s\_direct\_*}, \texttt{u\_p4s\_bound\_*},
\texttt{u\_p4t\_bound\_*}, \texttt{u\_p4u\_bound\_*},
\texttt{u\_p4o\_bound\_*}.

\item[$\mathcal{P}_{c}$]
\texttt{u\_p4s\_bound\_*}, \texttt{u\_p4t\_bound\_*},
\texttt{u\_p4u\_bound\_*}, \texttt{u\_p4o\_bound\_*}.

\item[$\mathcal{P}_{m}$]
\texttt{a\_p4k\_delivers\_*}, \texttt{a\_p4n\_delivers\_*},
\texttt{a\_p4o\_delivers\_*}, \texttt{a\_p4s\_delivers\_*}.
\end{description}

Table~\ref{tab:app-raw-projection} gives the projection; an AP also accepts a
raw observation equal to its canonical identifier.

\begin{table*}[t]
\caption{Complete raw-to-canonical projection. Parameterized identifiers
range over the sets defined above.}
\label{tab:app-raw-projection}
\centering
\scriptsize
\setlength{\tabcolsep}{3pt}
\renewcommand{\arraystretch}{1.02}
\begin{tabularx}{\linewidth}{
    @{}
    >{\centering\arraybackslash}p{0.05\linewidth}
    >{\raggedright\arraybackslash}p{0.19\linewidth}
    >{\raggedright\arraybackslash}X
    @{}
}
\toprule
\textbf{AP} & \textbf{Raw support class} & \textbf{Accepted raw evidence} \\
\midrule
$C_h$ & Harm context
& \texttt{d\_active\_harm\_context};
  \texttt{c\_domain\_}$d$ for $d\in\mathcal{D}_c$;
  \texttt{d\_active\_harm\_}$d$ for $d\in\mathcal{D}_h$. \\
$D_o$ & Objective source
& \texttt{d\_p4s\_source\_}$f$ and \texttt{d\_p4u\_bound\_}$f$
  for $f\in\mathcal{F}_o$. \\
\midrule
$U_h$ & Harmful objective & $\mathcal{H}\cup\mathcal{P}_h$. \\
$U_x$ & Execution request
& $\mathcal{H}\cup\mathcal{P}_x$;
  \texttt{u\_prompt\_how\_to};
  \texttt{u\_prompt\_requests\_specific\_parameter};
  \texttt{u\_prompt\_stepwise\_decomposition};
  \texttt{u\_prompt\_goal\_persistence};
  \texttt{u\_requests\_operational\_parameter}. \\
$U_e$ & Evasion tactic
& \texttt{u\_jailbreak\_intent\_prompt};
  \texttt{u\_jailbreak\_persona};
  \texttt{u\_prompt\_constraint\_bypass};
  \texttt{u\_prompt\_requests\_rephrase};
  \texttt{u\_requests\_encoded\_output};
  \texttt{u\_requests\_polarity\_inversion}. \\
$U_c$ & Continuation
& $\mathcal{P}_c$;
  \texttt{u\_prompt\_requests\_continuation};
  \texttt{u\_multi\_turn\_context\_ref}. \\
$U_{\uparrow}$ & Intent increase
& \texttt{u\_prompt\_escalation};
\texttt{u\_user\_escalates\_after\_partial}. \\
$U_o$ & Objective reference
& \texttt{u\_reference\_active\_harm\_details}. \\
\midrule
$A_s$ & Substantive content
& \texttt{a\_delivered\_content};
  \texttt{a\_harmful\_artifact\_generated};
  \texttt{a\_argumentative\_harm\_delivered};
  \texttt{a\_false\_claim\_delivered};
  \texttt{a\_targeted\_degradation\_delivered}. \\
$A_p$ & Procedural detail
& \texttt{a\_provides\_instructions};
  \texttt{a\_harmful\_action\_details\_delivered};
  \texttt{a\_zh\_concrete\_harm\_details\_delivered};
  \texttt{a\_compact\_harmful\_list\_delivered}. \\
$A_m$ & Method specificity
& $\mathcal{P}_m$; \texttt{a\_names\_specific\_method};
\texttt{a\_operational\_specificity}. \\
$A_e$ & Action enablement
& \texttt{a\_provides\_workaround};
  \texttt{a\_harmful\_artifact\_generated};
  \texttt{a\_confirms\_feasibility}. \\
$A_t$ & Targeted content
& \texttt{a\_argumentative\_harm\_delivered};
  \texttt{a\_false\_claim\_delivered};
  \texttt{a\_targeted\_degradation\_delivered}. \\
$A_{\ell}$ & Partial assistance
& \texttt{a\_partial\_compliance};
  \texttt{a\_partial\_refusal};
  \texttt{a\_confirms\_feasibility}. \\
$A_r$ & Effective refusal
& \texttt{a\_effective\_refusal};
  \texttt{a\_full\_refusal};
  \texttt{a\_refused\_safety};
  \texttt{a\_p4k\_effective\_refusal};
  \texttt{a\_p4k\_refusal\_act}. \\
\bottomrule
\end{tabularx}
\end{table*}

Let \(\mathcal{S}(a)\) denote the support set in
Table~\ref{tab:app-raw-projection}, including the canonical identifier itself.
The projection records
\[
\begin{aligned}
\operatorname{Prov}_t(a)
&=\bigl\{q\in R_t\setminus\mathcal{R}_{\mathrm{end}}
          \mid q\in\mathcal{S}(a)\bigr\},\\
a\in A_t
&\Longleftrightarrow \operatorname{Prov}_t(a)\neq\varnothing.
\end{aligned}
\]
The excluded set is
\begin{ltlformula}
\begin{aligned}
\mathcal{R}_{\mathrm{end}}=\{\,&
\texttt{u\_attack\_state},
\texttt{u\_attack\_observed},\\
&\texttt{u\_attack\_signal},
\texttt{u\_content\_risk},\\
&\texttt{u\_conversation\_intent},\\
&\texttt{a\_response\_unsafe},\\
&\texttt{a\_response\_unsafe\_observed},\\
&\texttt{a\_harmful\_delivery},
\texttt{a\_delivery\_evidence},\\
&\texttt{a\_strong\_delivery},
\texttt{a\_operational\_delivery}\}.
\end{aligned}
\end{ltlformula}
Consequently,
\(
\pi(R_t)=\pi(R_t\setminus\mathcal{R}_{\mathrm{end}})
\).

\subsection{Typed Temporal Constraint Library}
\label{app:constraint-library}

\noindent\textbf{Typed intermediate representation.}
An event is \(e=(\rho,P,Q,N)\), with role \(\rho\), required APs \(P\),
alternative APs \(Q\), and disqualifiers \(N\). It matches \(A_t\) iff
\begin{ltlformula}
\begin{aligned}
&\mathsf{Match}(e,A_t)=1
\Longleftrightarrow\\
&P\subseteq A_t
\land
\bigl(Q=\varnothing\lor Q\cap A_t\neq\varnothing\bigr)
\land
N\cap A_t=\varnothing.
\end{aligned}
\end{ltlformula}
We write \(\mathsf{Match}_e\) for this predicate in the compiled formula.
Role validation permits user events to draw on user/context APs, reply events on reply/context/history APs, and round events on context/history APs. We identify each constraint $C_i$ with its compiled formula $\phi_i$, 
writing $C_i \in \Phi$ and $\operatorname{fp}(C_i) = \operatorname{fp}(\phi_i)$ 
interchangeably.

A typed constraint \(C_i=\{(e_{ij},\delta_{ij})\}_j\) assigns relative offsets;
the library uses only \(\delta_{ij}\in\{0,1\}\) and compiles deterministically:
\begin{ltlformula}
\begin{aligned}
\phi_i
=
\mathbf{F}\!\left(
\bigwedge_j \mathbf{X}^{\delta_{ij}}
\mathsf{Match}_{e_{ij}}
\right).
\end{aligned}
\end{ltlformula}
Events sharing an offset use the same transaction. Define
\[
Q_A
\triangleq
A_s
\land(A_p\lor A_m\lor A_e\lor A_t)
\land\neg A_r.
\]
\(Q_A\) is a Boolean abbreviation and is not an additional AP.

\smallskip
\noindent\textbf{Constraint inventory.}
Table~\ref{tab:app-constraint-inventory} lists the ten compiled rules; four use
offset \((0,0)\) and six use \((0,1)\).

\begin{table*}[t]
\centering
\scriptsize
\caption{Released typed constraints; $Q_A$ is defined above.}
\label{tab:app-constraint-inventory}
\setlength{\tabcolsep}{4pt}
\renewcommand{\arraystretch}{1.05}
\begin{tabularx}{\textwidth}{@{}p{0.075\textwidth}p{0.18\textwidth}p{0.31\textwidth}p{0.08\textwidth}X@{}}
\toprule
\textbf{Rule} & \textbf{Implementation identifier} & \textbf{Compiled formula}
& \textbf{Offsets} & \textbf{Forbidden relation} \\
\midrule
 $\phi_1$ ($R_1$) & \texttt{P13F-SAME-TURN-PROCEDURE}
& $\displaystyle \mathbf{F}\!\left(U_h\land U_x\land A_s\land A_p\land\neg A_r\right)$
& $(0,0)$
& Harmful execution request with procedural detail and no refusal. \\
 $\phi_2$ ($R_2$) & \texttt{P13F-CONTEXT-SPECIFICITY}
& $\displaystyle \mathbf{F}\!\left(C_h\land A_s\land A_m\land\neg A_r\right)$
& $(0,0)$
& Active harm context with method-specific assistance. \\
 $\phi_3$ ($R_3$) & \texttt{P13F-SAME-TURN-ENABLEMENT}
& $\displaystyle \mathbf{F}\!\left(\begin{gathered}U_h\land U_x\land A_s\\[-1pt]\land(A_e\lor A_t)\land\neg A_r\end{gathered}\right)$
& $(0,0)$
& Harmful request with enabling or targeted content. \\
 $\phi_4$ ($R_4$) & \texttt{P13F-EVASIVE-ASSISTANCE}
& $\displaystyle \mathbf{F}\!\left(U_h\land U_e\land Q_A\right)$
& $(0,0)$
& Evasive harmful objective with detailed assistance. \\
\midrule
 $\phi_5$ ($R_5$) & \texttt{P13F-CONTINUATION}
& $\displaystyle \mathbf{F}\!\left(\begin{gathered}U_h\land\\[-1pt]\mathbf{X}\!\left(U_c\land Q_A\right)\end{gathered}\right)$
& $(0,1)$
& Continued objective with next-turn detailed assistance. \\
 $\phi_6$ ($R_6$) & \texttt{P13F-ESCALATION}
& $\displaystyle \mathbf{F}\!\left(\begin{gathered}U_h\land\\[-1pt]\mathbf{X}\!\left(U_{\uparrow}\land Q_A\right)\end{gathered}\right)$
& $(0,1)$
& Escalated objective with next-turn detailed assistance. \\
 $\phi_7$ ($R_7$) & \texttt{P13F-REFUSAL-BYPASS}
& $\displaystyle \mathbf{F}\!\left(\begin{gathered}A_r\land\\[-1pt]\mathbf{X}\!\left((U_c\lor U_o)\land Q_A\right)\end{gathered}\right)$
& $(0,1)$
& Refusal followed by renewed request and assistance. \\
 $\phi_8$ ($R_8$) & \texttt{P13F-PARTIAL-ESCALATION}
& $\displaystyle \mathbf{F}\!\left(\begin{gathered}A_{\ell}\land\\[-1pt]\mathbf{X}\!\left((U_c\lor U_{\uparrow})\land Q_A\right)\end{gathered}\right)$
& $(0,1)$
& Partial assistance followed by continuation/escalation. \\
 $\phi_9$ ($R_9$) & \texttt{P13F-OBJECTIVE-REFERENCE}
& $\displaystyle \mathbf{F}\!\left(\begin{gathered}D_o\land\\[-1pt]\mathbf{X}\!\left((U_c\lor U_o)\land Q_A\right)\end{gathered}\right)$
& $(0,1)$
& Established objective resumed with assistance. \\
 $\phi_{10}$ ($R_{10}$) & \texttt{P13F-REPEATED-ASSISTANCE}
& $\displaystyle \mathbf{F}\!\left(\begin{gathered}(C_h\land Q_A)\land\\[-1pt]\mathbf{X}\!\left(C_h\land Q_A\right)\end{gathered}\right)$
& $(0,1)$
& Detailed assistance in consecutive harm-context turns. \\
\bottomrule
\end{tabularx}
\end{table*}

\subsection{Witness Construction and Replay}
\label{app:witness-replay}

For constraint \(C_i\), let \(\Delta_i=\max_j\delta_{ij}\). Its matching
anchors on \(w=A_0\cdots A_T\) are
\[
K_i(w)=\left\{k\in\{0,\ldots,T-\Delta_i\}:\
\bigwedge_j e_{ij}(A_{k+\delta_{ij}})=1\right\}.
\]
If nonempty, choose \(k_i^\star=\min K_i(w)\). For event \(e_{ij}\), let
\[
M_{ij}
=
P_{ij}\cup\bigl(Q_{ij}\cap A_{k_i^\star+\delta_{ij}}\bigr)
\]
be its realizing APs. The witness is
\[
\begin{aligned}
\omega_i=\bigl(&i,k_i^\star,
\{(\rho_{ij},k_i^\star+\delta_{ij},M_{ij},
\{\operatorname{Prov}_{k_i^\star+\delta_{ij}}(a)
\}_{a\in M_{ij}})\}_j,\\
&\operatorname{fp}(C_i),
\operatorname{dig}(w)\bigr).
\end{aligned}
\]
The constraint fingerprint hashes the typed events, offsets, and schema; the
trace digest hashes schema identifiers and sorted APs.

Verification recomputes the earliest match, checks roles, turns, APs, both
digests, and every provenance set; any changed field fails replay.

\subsection{Formal Monitoring Guarantees}
\label{app:monitoring-proofs}

The guarantees below hold at the canonical-trace boundary under: (A1)
deterministic activation for a fixed prefix; (A2) causal activation; (A3)
removal of $\mathcal{R}_{\mathrm{end}}$ before projection; and (A4) the
fixed-window existential constraint forms of
Sec.~\ref{subsec:specguard-monitoring}.

Fix $C$ with automaton
$\mathcal{M}_C=(Q,2^{\mathcal A},\delta,q_0,F)$ and trace
$w=A_0\cdots A_T$. Let
\begin{ltlformula}
\begin{aligned}
&\operatorname{Run}_C(w)=(s_0,\ldots,s_{T+1}),\quad
s_0=q_0,\quad s_{j+1}=\delta(s_j,A_j),\\
&\qquad
\pi_C(w)=(A_0\cap\mathcal A_C)\cdots(A_T\cap\mathcal A_C),
\end{aligned}
\end{ltlformula}
where $\mathcal A_C$ is the AP set referenced by $C$.

\smallskip
\noindent\textbf{Proof of Theorem~\ref{thm:finite-monitor-correctness}.}
For a same-turn rule $\mathbf F(B)$, induction on the input position gives
$s_{j+1}=q_V$ iff some $k\leq j$ satisfies $A_k\models B$. For an
adjacent-turn rule $\mathbf F(H\land\mathbf XG)$, the automaton stores whether
the latest transaction satisfies $H$; it enters the absorbing state exactly
when the next transaction satisfies $G$. Hence both released templates satisfy
\[
w\models_f\phi_C
\quad\Longleftrightarrow\quad
w\in L(\mathcal M_C).
\]
\(\square\)

\smallskip
\noindent\textbf{Proposition 1 (Prefix consistency).}
\[
D_{\leq t}=D'_{\leq t}
\quad\Longrightarrow\quad
w(D)_{\leq t}=w(D')_{\leq t}.
\]
\emph{Proof.} Equal prefixes give equal causal views; (A1)--(A2) therefore
give equal evidence, projection, and trace positions. \(\square\)

\smallskip
\noindent\textbf{Proposition 2 (Runtime--audit equivalence).}
If the runtime candidate reply is the reply in the audited prefix, then
\begin{ltlformula}
\begin{aligned}
\mathsf{Alarm}_{\rm run}(D_{\leq t})
=\mathsf{Alarm}_{\rm audit}(D_{\leq t})
=\mathbf 1\!\left[\exists C\in\Phi:
w(D_{\leq t})\models_f\phi_C\right].
\end{aligned}
\end{ltlformula}
\emph{Proof.} Both modes construct the same causal trace by Proposition 1;
automaton correctness gives the same acceptance for every constraint.
\(\square\)

\smallskip
\noindent\textbf{Proposition 3 (Constraint-projection invariance).}
For equal-length traces,
\begin{ltlformula}
\begin{aligned}
&\pi_C(w)=\pi_C(w')
\Longrightarrow
\operatorname{Run}_C(w)=\operatorname{Run}_C(w')\\
&\Longrightarrow
\bigl(w\in L(\mathcal M_C)\Leftrightarrow
w'\in L(\mathcal M_C)\bigr).
\end{aligned}
\end{ltlformula}
\emph{Proof.} Every transition predicate references only $\mathcal A_C$;
induction on positions gives identical state sequences. \(\square\)

\smallskip
\noindent\textbf{Proposition 4 (Endpoint non-interference).}
For $X,Y\subseteq\mathcal R_{\mathrm{end}}$,
\[
\pi\bigl((R_t\setminus\mathcal R_{\mathrm{end}})\cup X\bigr)
=
\pi\bigl((R_t\setminus\mathcal R_{\mathrm{end}})\cup Y\bigr).
\]
Thus excluded endpoint predicates cannot change the trace, automaton run, or
verdict. \emph{Proof.} (A3) discards $X$ and $Y$ before support mapping.
\(\square\)

\smallskip
\noindent\textbf{Proposition 5 (Verdict--witness correspondence).}
For every finite trace $w$,
\begin{ltlformula}
\begin{aligned}
&\mathsf{Alarm}(w)=1
\Longleftrightarrow
\exists C_i\in\Phi,\exists\omega_i:\\
&\ w\models_f\phi_i\land
\operatorname{Verify}(\omega_i,C_i,w)=1.
\end{aligned}
\end{ltlformula}
\emph{Proof.} Verification replays the matched events, offsets, AP support,
and digests, giving soundness. If a rule accepts, the finite anchor set
$K_i(w)$ is nonempty; its earliest element generates a witness that replay
accepts, giving completeness. \(\square\)

\smallskip
\noindent\textbf{Proposition 6 (Extension monotonicity).}
\begin{ltlformula}
\begin{aligned}
&w\in L(\mathcal M_C)\Longrightarrow w\cdot v\in L(\mathcal M_C),
\qquad\\
&\mathsf{Alarm}(w)=1\Longrightarrow\mathsf{Alarm}(w\cdot v)=1.
\end{aligned}
\end{ltlformula}
\emph{Proof.} The accepting state is absorbing, and the fixed-window
existential witness is unchanged by appending a suffix. \(\square\)

\smallskip
\noindent\textbf{Proposition 7 (Monitoring complexity).}
With fixed-width transition predicates, excluding AP activation and text
storage, a runtime update costs $O(|\Phi|)$ time and state; replaying $T+1$
turns costs $O(T|\Phi|)$. \emph{Proof.} Each rule uses a constant-size
automaton and one constant-width transition per turn. \(\square\)

\smallskip
\noindent\textbf{Boundary.}
These results establish deterministic acceptance, mode equivalence,
non-interference, replayability, monotonicity, and complexity after the
canonical trace is instantiated. They do not claim perfect natural-language
fact production; missing or spurious APs are evaluated separately.

\end{document}